\documentclass{article}

\usepackage{arxiv}
\usepackage[authoryear]{natbib}

\usepackage[utf8]{inputenc} 
\usepackage[T1]{fontenc}    
\usepackage{hyperref}       
\usepackage{url}            
\usepackage{booktabs}       
\usepackage{amsfonts}       
\usepackage{nicefrac}       
\usepackage{microtype}      
\usepackage{lipsum}
\usepackage{graphicx}
\usepackage{siunitx}
\usepackage{amsmath}
\usepackage{graphicx}
\usepackage{amssymb}
\usepackage{minted}
\usepackage{subcaption}
\usepackage{multirow}
\usepackage{booktabs}
\usepackage{algorithm}   
\usepackage{algpseudocode} 
\usepackage{tabularx}
\providecommand{\R}{{\mathbb R}}
\providecommand{\N}{{\mathbb N}}
\providecommand{\E}{{\mathbb E}}

\DeclareMathOperator*{\argmin}{argmin}

\providecommand{\BA}{{\mathbf{A}}}
\providecommand{\BB}{{\mathbf{B}}}
\providecommand{\BD}{{\mathbf{D}}}
\providecommand{\Bz}{{\mathbf{z}}}
\providecommand{\Bx}{{\mathbf{x}}}
\providecommand{\By}{{\mathbf{y}}}
\providecommand{\BI}{{\mathbf{I}}}
\providecommand{\Bu}{{\mathbf{u}}}
\providecommand{\BJ}{{\mathbf{J}}}
\providecommand{\BK}{{\mathbf{K}}}

\providecommand{\Btheta}{{\boldsymbol{\theta}}}
\providecommand{\reg}{{\mathcal J}}
\providecommand{\Bla}{{\boldsymbol{\lambda}}}
\usepackage{upgreek}

\newcommand{\Veta}       {\boldsymbol{\upeta}} %
\newcommand{\Vmu}       {\boldsymbol{\upmu}} %

\usepackage{mathtools}

\usepackage{amsthm}
\newtheorem{theorem}{Theorem}
\newtheorem{proposition}{Proposition}

\title{Deep Image Prior for Computed Tomography Reconstruction}

\author{
  Simon Arridge \\
 Department of Computer Science \\
 University College London \\
   \And
  Riccardo Barbano \\
 Department of Computer Science \\
 University College London \\
  \AND
  Alexander Denker \\
 Department of Computer Science \\
 University College London \\
  \And
  \v{Z}eljko Kereta \\
 Department of Computer Science \\
 University College London \\
}

\begin{document}

\maketitle
\begin{abstract}
We present a comprehensive overview of the Deep Image Prior (DIP) framework and its applications to image reconstruction in computed tomography.
Unlike conventional deep learning methods that rely on large, supervised datasets, the DIP exploits the implicit bias of convolutional neural networks and operates in a fully unsupervised setting, requiring only a single measurement, even in the presence of noise.
We describe the standard DIP formulation, outline key algorithmic design choices, and review several strategies to mitigate overfitting, including early stopping, explicit regularisation, and self-guided methods that adapt the network input.
In addition, we examine computational improvements such as warm-start and stochastic optimisation methods to reduce the reconstruction time. 
The discussed methods are tested on real $\mu$CT measurements, which allows examination of trade-offs among the different modifications and extensions.
\end{abstract}


\section{Introduction}
Image reconstruction plays a critical role across many scientific disciplines and technological domains.
The range of applications extends from computer vision tasks such as image denoising, deblurring, or super-resolution, to complex scientific challenges in which only indirect observations, acquired by dedicated sensing devices, are available.
Examples of the latter include problems in microscopy, astronomy, geophysics, and medical imaging.
The common thread between these applications is the inverse nature of the problem: we seek to recover an image from corrupted or incomplete indirect observations.
Inverse problems in imaging are mathematically formulated as recovering or approximating $\Bx$, the unknown ground truth, given measured data $\By^\delta$, from an equation
\begin{align}
    \By^\delta = \BA(\Bx) + \Veta,
\end{align}
where $\BA: \R^n \to \R^m$ is the forward operator, describing the physics of the data acquisition or distortion process, and $\Veta$ is the measurement noise.
This chapter focuses on discrete linear inverse problems, where noisy data $\By^\delta$ arises from linear (or linearised) finite projections, and the forward operator $\BA \in \R^{n \times m}$ is a matrix relating measurements and images.

Inverse problems are difficult to solve because they are often ill-posed in the sense of Hadamard.
This can manifest as non-existence or non-uniqueness of solutions, or through severe instability, whereby small perturbations in the measured data can result in large errors in the reconstructed solutions \citep{tikhonov1977solutions,Engl1996}.
These issues bear a special significance for medical imaging applications, due to their implications on quality and validity of a diagnosis or course of treatment.
Inverse problems in imaging have traditionally been addressed through \emph{knowledge-based} approaches that assume full or partial access to an approximate physical model of the measurement process.
For example, image reconstruction in computed tomography (CT) has been addressed with analytical and iterative reconstructions methods \citep{natterer2001mathematics}.
For 2D-CT the filtered back projection (FBP) \citep{kak2001principles} is the standard analytical approach, whose susceptibility to noise and sparsity of the measurement data has resulted in the development of projection-based iterative reconstruction algorithms, e.g., the algebraic reconstruction technique \citep{gordon1970algebraic}, or the simultaneous iterative reconstruction technique \citep{andersen1984simultaneous}.
Building on these foundations, variational regularisation emerged as a comprehensive reconstruction framework \citep{scherzer2009variational}.
It reformulates the reconstruction as an optimisation problem
\begin{align}
    \label{eq:var_obj}
    \hat{\Bx} \in \argmin_{\Bx \in \R^n} \ell(\BA \Bx, \By^\delta) + \alpha \mathcal{J}(\Bx),
\end{align}
that combines the data fidelity $\ell:\R^m \times \R^m \to \R_{\ge 0}$ with a regularising image prior $\mathcal{J}:\R^n \to \R_{\ge 0}$, which promotes desired image properties while simultaneously stabilising the reconstruction process. 
In spite of their success, knowledge-based approaches have several shortcomings, such as reliance on hand-crafted regularisers, need for domain expertise or the use of overly simplistic models of structures found in natural or medical images.
As a result, recent decades have seen a paradigm shift with a growing flood of deep learning–based approaches that have revolutionised image reconstruction in medical imaging.
A wide spectrum of these methods is discussed by \cite{arridge2019solving,ongie2020deep,hertrich2025learning,habring2024neural}, including post-processing approaches, end-to-end mappings from measurements to images, and unrolled networks, just to name a few.

The vast majority of these methods critically depend on the availability of large, paired datasets for training.
Considerable effort has been devoted to developing approaches that mitigate this reliance on paired training data.
In particular, a number of self-supervised frameworks were proposed; see the recent review by \cite{tachella2026self}.
In this context, self-supervision refers to approaches that assume access only to a set $\{\By^{(i)}\}_i$ of (noisy) measured data for training, without requiring ground truth images. 
For example, Noise2Noise \citep{lehtinen2018noise2noise}, eliminates the need for clean reference data by training on pairs of independently corrupted measurements.
Extensions such as Noise2Void \citep{krull2019noise2void} and Noise2Self \citep{batson2019noise2self} further reduce data requirements using only single noisy images.
While these methods were originally developed for denoising tasks, they can be extended to linear inverse problems by employing them as post-processing steps applied to an initial reconstruction.
In contrast, Noise2Inverse \citep{hendriksen2020noise2inverse} adapts these principles directly to linear inverse problems by exploiting redundancies in the measurement domain.
Furthermore, equivariant imaging \citep{chen2021equivariant} exploits symmetries of the forward operator $\BA$ by enforcing equivariance constraints during training.

This work focuses on the Deep Image Prior (DIP), an unsupervised method proposed by \cite{ulyanov2018deep}, which requires only a single (noisy) measurement for reconstruction.
The principal idea of DIP is to reparametrise the image to be reconstructed as the output of a convolutional neural network (CNN).
The parameters of the network are then optimised by solving a minimisation problem analogous to \eqref{eq:var_obj}, but with respect to the network parameters.
The regularising effects are imposed by the fixed convolutional architecture, early stopping, and other factors described in the following sections.
Owing to this formulation, the framework is fully unsupervised, requiring no dataset of ground truth images.

This chapter gives an overview of key extensions to the DIP, with a focus on their use for CT reconstruction.
We present strategies for mitigating overfitting, such as incorporating additional regularisation or improved early stopping rules, as well as strategies for reducing the computational costs, such as warm-start approaches or stochastic gradient descent.
We also explore connections to related concepts, such as self-supervised denoising and test-time-training \citep{sun2020test}.

The rest of this chapter is structured as follows. 
We begin by introducing the general DIP framework in Section \ref{sec:general_framework}. 
In Section \ref{sec:theoretical_approaches}, we briefly cover theoretical advances in understanding the behaviour of the DIP, particularly looking at the under- and over-parametrised regimes. 
We discuss methods to prevent overfitting and reduce the computational costs in Section \ref{sec:overfitting} and Section \ref{sec:computational_cost}, respectively. 
Finally, we compare these different approaches and extensions on a real-world $\mu$CT dataset in Section \ref{sec:numerical_experiments}. 
We draw conclusions in Section \ref{sec:conclusion}.  

\section{General Framework}
\label{sec:general_framework}
We start by describing the general DIP framework as developed by \cite{ulyanov2018deep}. 
For this, we denote the neural network as $f: \R^n \times \R^p \to \R^n$ with network input $\Bz \in \R^n$ and parameters $\Btheta \in \R^p$. 
The DIP optimises the network parameters $\Btheta$ for a fixed network input $\Bz$ such that the network output is consistent to the noisy measurements, leading to the optimisation problem
\begin{align}
    \hat{\Btheta} \in \argmin_{\Btheta \in \R^p} \left\{ \ell(\BA f_\Btheta(\Bz), \By^\delta) \coloneqq \frac12\|\BA f_\Btheta(\Bz)-\By^\delta\|_2^2 \right\}. \label{eq:dip_vanilla}
\end{align}
The reconstructed image is then retrieved by evaluating the network at the optimised parameters as $\hat{\Bx} = f_{\hat{\Btheta}}(\Bz)$.
We will refer to this formulation as the \textit{vanilla DIP}. 
The optimisation procedure is analogous to that of the variational objective in \eqref{eq:var_obj}, but with two key differences. 
First, the optimisation is performed over the network parameters rather than directly in the image space.
Second, the vanilla DIP formulation \eqref{eq:dip_vanilla} does not incorporate the regulariser $\mathcal{J}$.
The optimisation process is summarised in Algorithm \ref{alg:vanilla_dip}.
In the following, we use the mean squared error as the data fidelity term $\ell$, although alternative choices are also possible.
In the following sections we will discuss several variations and extensions of the vanilla DIP, each accompanied by an algorithm in a format consistent with Algorithm~\ref{alg:vanilla_dip}.

The DIP framework builds on the empirical observation that although neural networks can often fit any signal, some architectures, particularly CNNs, capture natural image structures faster than noise. 
This behaviour is often linked to the spectral bias of CNNs, whereby low-frequency components (associated with natural image structures) are captured early in the optimisation process, while high-frequency components (associated with noise) are learned later.
This tendency can be interpreted as a form of implicit regularisation, often referred to as regularisation by architecture.
In the example in Figure \ref{fig:dip_reco_iters} we can observe that after $100$ iterations the general outline of the walnut is already visible. Finer details gradually emerge in subsequent iterations.
\begin{figure}[t]
 \begin{subfigure}[t]{.19\textwidth}
    \includegraphics[width=1.0\linewidth]{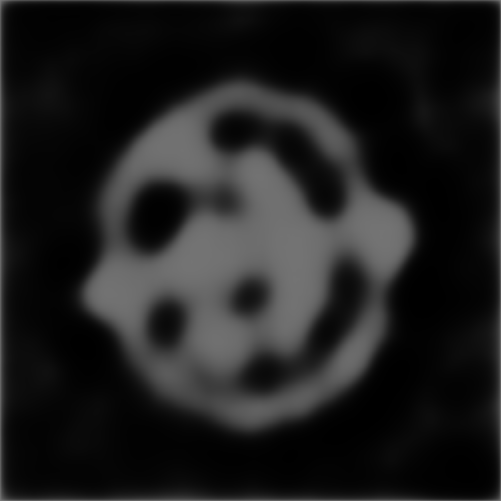}%
     \captionsetup{justification=centering}
    \caption*{Iteration $100$}
\end{subfigure}%
\hfill
 \begin{subfigure}[t]{.19\textwidth}
    \includegraphics[width=1.0\linewidth]{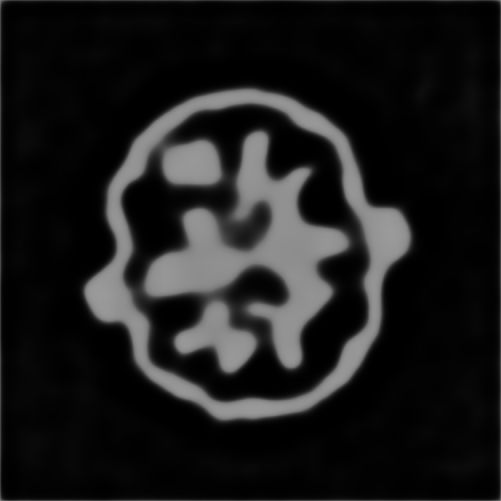}%
         \captionsetup{justification=centering}
    \caption*{Iteration $200$}
\end{subfigure}%
\hfill
 \begin{subfigure}[t]{.19\textwidth}
    \includegraphics[width=1.0\linewidth]{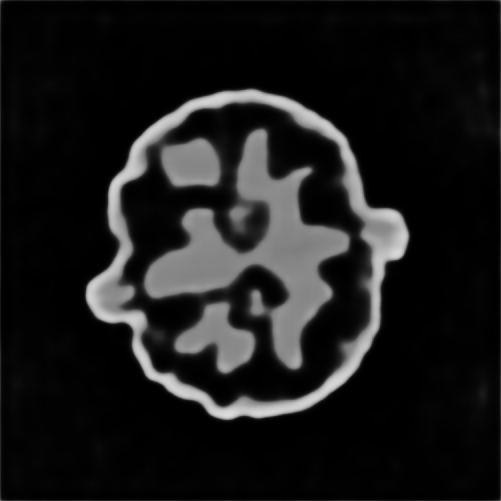}%
         \captionsetup{justification=centering}
    \caption*{Iteration $400$}
\end{subfigure}%
\hfill
 \begin{subfigure}[t]{.19\textwidth}
    \includegraphics[width=1.0\linewidth]{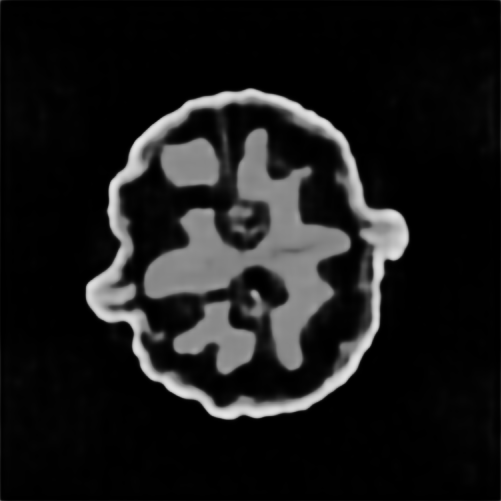}%
         \captionsetup{justification=centering}
    \caption*{Iteration $600$}
\end{subfigure}%
\hfill
 \begin{subfigure}[t]{.19\textwidth}
    \includegraphics[width=1.0\linewidth]{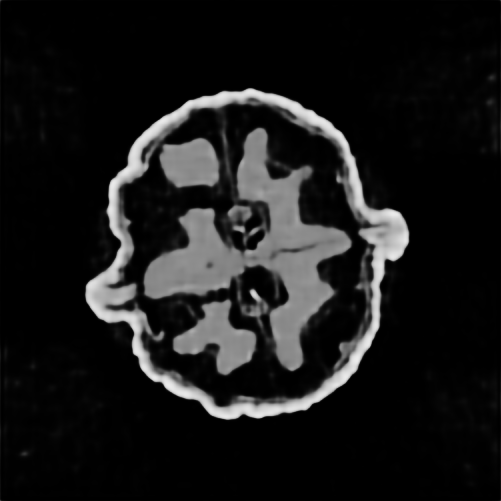}%
         \captionsetup{justification=centering}
    \caption*{Iteration $1000$}
\end{subfigure}%
\caption{Vanilla DIP on the $\mu$CT walnut with $60$ angles and $128$ detector pixels. We show the DIP reconstruction at different iterations of the optimisation process. Details in Section \ref{sec:numerical_experiments}.}
\label{fig:dip_reco_iters}
\end{figure}

\begin{algorithm}[t]
\caption{Deep Image Prior}
\label{alg:vanilla_dip}
\begin{algorithmic}[1]
\Require Forward operator $\mathbf{A}$, observation $\mathbf{y}^\delta$, maximum number of iteration $N_\text{max}$, model input $\Bz$, model architecture $f_\Btheta$
\For{$k=1$ to $N_\text{max}$}
    \State $\tilde{\Bx} =  f_\Btheta(\Bz)$
    \State $L(\Btheta) = \frac{1}{2} \| \BA \tilde{\Bx} - \By^\delta \|_2^2 $
    \State Take gradient step w.r.t.~$\Btheta$ using $\nabla L(\Btheta)$
\EndFor
\State \textbf{return} $f_\Btheta(\Bz)$
\end{algorithmic}
\end{algorithm}
However, ensuring the reconstruction is of the best quality is not straightforward, as optimisation must be stopped before the network starts fitting the noise. 
This is illustrated in Figure \ref{fig:vanilla_dip}, which shows a decrease in the loss function, quantifying data fidelity, over the entire optimisation trajectory. 
The best reconstruction, as measured by PSNR, is obtained relatively early in the optimisation (see the dashed vertical green line in Figure \ref{fig:vanilla_dip}).
We also show the reconstruction obtained by early stopping the optimisation process (see the dashed vertical red line in Figure~\ref{fig:vanilla_dip}) using the method by \cite{wang2021early}, described in Section~\ref{sec:early_stopping}. This early stopped reconstruction is better than the final reconstruction (dashed vertical blue line in Figure~\ref{fig:vanilla_dip}), but there is still a gap to the peak PSNR.
There is an even larger discrepancy between the highest PSNR attainable with DIP and the PSNR of the final reconstruction (using fully optimised network parameters).
This form of overfitting poses a major challenge as the number of optimisation steps required to obtain the best reconstruction is not known a priori and strongly depends on the task, the neural network architecture, the noise level in the measurement data, and the used optimisation algorithm.
Avoiding or mitigating overfitting has spurred significant research in the DIP community.
An overview of these methods is provided in Section \ref{sec:overfitting}.

Note that overfitting to noise is not inherent to the DIP, it is a direct consequence of the fact that the optimisation function \eqref{eq:dip_vanilla} contains only the data fidelity term. 
If no regularisation functional is included in the variational formulation \eqref{eq:var_obj}, overfitting will also occur, and early stopping must be applied to achieve a regularisation method, which then leads to the classical Landweber algorithm \citep{landweber1951iteration}.

\begin{figure}[t]
    \centering
\includegraphics[width=1.0\linewidth]{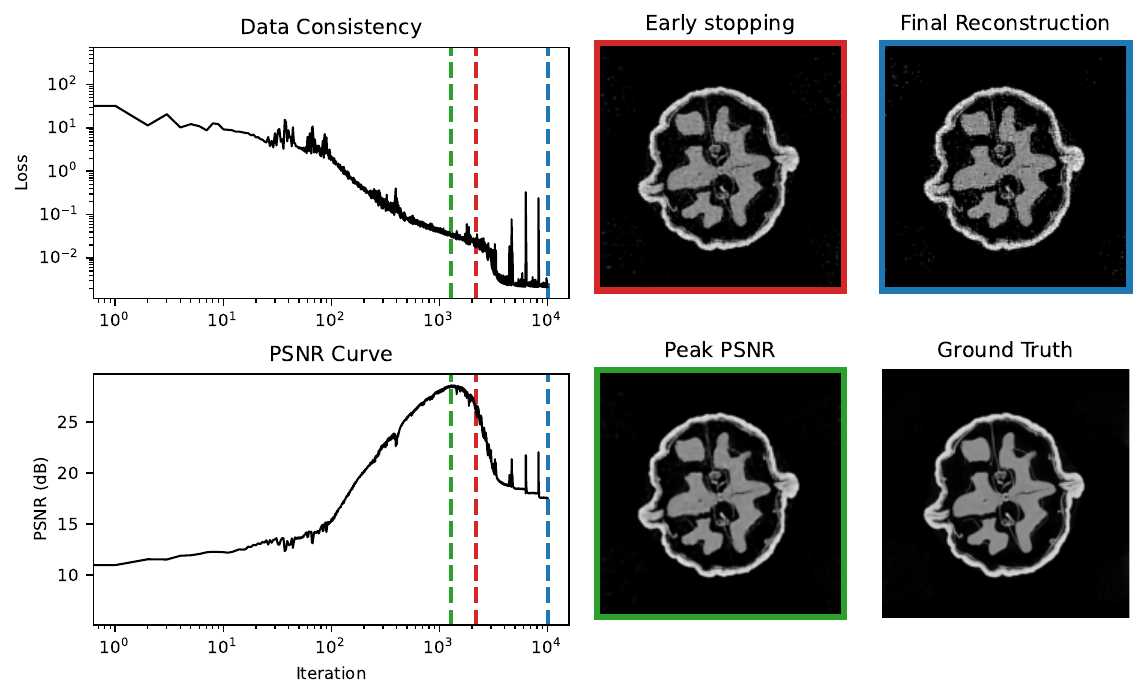}
    \caption{Vanilla DIP on the $\mu$CT walnut with $60$ angles and $128$ detector pixels. For details see the description in Section \ref{sec:numerical_experiments}.} 
    \label{fig:vanilla_dip}
\end{figure}

\subsection{Algorithmic Design Choices}\label{sec:algo_design_choices}

There are several important algorithmic choices which influence the performance of the DIP: the network input $\Bz$, the initialisation of the network's parameters $\Btheta_0$, and the design of the architecture of the network.

\textbf{Network Input}:
In the original DIP work, the authors use both random noise $\Bz \sim \mathcal{N}(0, \sigma^2 \BI_m)$ as well as deterministic ``mesh-grid'' inputs (see Figure~22 in \cite{ulyanov2018deep}). 
The design of the input can be characterised along two main axes. 
The first concerns its relevance to the reconstruction task: the input can be uninformative (e.g., a random noise mask) or informative (e.g., an approximate reconstruction obtained with classical methods).
Studies have shown that using informative inputs can both improve the quality of reconstruction and reduce computation time. 
Examples include using a rough initial reconstruction such as the FBP \citep{barbano2022educated,tachella2021neural}, or leveraging a related reference image \citep{zhao2020reference}.
The second axis concerns whether the input is updated during optimisation: it can either remain fixed or be updated jointly with the network parameters.
Here, several recent works have proposed to update the input during optimisation \citep{liang2025analysis,alkhouri2024image}. These approaches will be discussed in detail later, see in particular Algorithm~\ref{alg:self_guided_dip} and Algorithm~\ref{alg:aSeqDIP}.

\textbf{Network Architecture}: Architectural design strongly influences both role the regularising effect and the tendency to fit high-frequency structures, such as noise or streaking artefacts that arise in CT reconstruction.
While most DIP implementations employ a standard convolutional autoencoder, in particular the popular U-Net architecture \citep{ronneberger2015u}, other choices have been explored \citep{korkmaz2022unsupervised,schmidt2021deep,dittmer2020regularization}.
For example, \cite{heckel2018deep} observe that even simple network architectures can achieve good performance.
An empirical comparison lies beyond the scope of this work and our numerical studies use the same U-Net architecture across methods.

\textbf{Initialisation of Parameters}:
In most DIP studies, the parameters $\Btheta$ are initialised randomly, using standard deep learning initialisation strategies.
However, the choice of initialisation influences the optimisation trajectory and thus the reconstructed image. \cite{barbano2022educated} explore strategies for improved initialisation.
They ``warm-start'' the network $f_\Btheta$ on a synthetic paired dataset using a supervised procedure, and use the resulting parameters to initialise the subsequent (unsupervised) reconstruction task.
This approach aims to reduce reconstruction time, since each DIP reconstruction would otherwise require training the entire network from scratch for reach new problem

\section{Theoretical Approaches}\label{sec:theoretical_approaches}

In this section, we discuss several modes of analysis of the DIP. We begin with the analytic deep prior (APD) in Section \ref{sec:adp}, followed by a discussion of the under-parametrised DIP in Section \ref{sec:deep_decoder}. We then turn to the neural tangent kernel perspective for sufficiently over-parametrised DIPs in Section \ref{sec:linearised_dip}.

\subsection{Analytic Deep Prior}\label{sec:adp}
The empirical success of the DIP is typically attributed to the implicit bias of the network architecture towards natural images. 
However, the reliance on large, over-parametrised networks makes establishing rigorous theoretical guarantees difficult.
To bridge this gap, a line of research focuses on \textit{regularisation by architecture}.
In this context, \cite{dittmer2020regularization} introduce the ADP. 
The core intuition of the ADP is to replace the standard black-box network with a specialised architecture that can be explicitly represented as the minimisation of some functional, leveraging methods such as LISTA \citep{gregor2010learning} or unrolled proximal gradient schemes \citep{adler2018learned,jiu2021deep}.
In its discrete form, the ADP reformulates the DIP \eqref{eq:dip_vanilla} by framing it as a bilevel optimisation problem
\begin{align}
    \label{eq:adp1}
    &\min_{\BB \in \R^{m\times n}} \frac{1}{2} \| \BA \hat{\Bx}(\BB) - \By^\delta \|_2^2 \\ 
    \text{subject to} &\quad \hat{\Bx}(\BB) \coloneqq \argmin_\Bx \frac{1}{2} \| \BB \Bx - \By^\delta \|_2^2 + \alpha \mathcal{J}(\Bx), \label{eq:adp2}
\end{align}
where $\mathcal{J}: \R^n \to \R \cup \{ \infty \}$ is a convex regularisation functional and $\alpha>0$ is the regularisation strength. In this formulation, the matrix $\BB$ represents the learnable parameters. The lower-level problem \eqref{eq:adp2} induces a proximal gradient descent type architecture, where the $l$-th layer of the network architecture can be formulated as
\begin{align}
    \Bx^{(l+1)} &= \text{prox}_{\alpha \lambda \mathcal{J}}(\Bx^{(l)} - \lambda \BB^T (\BB \Bx^{(l)} - \By^\delta)) \\
    &= \text{prox}_{\alpha \lambda \mathcal{J}}((\mathbf{I} - \lambda \BB^T \BB) \Bx^{(l)} - \lambda \BB^T \By^\delta).
\end{align}
This expression can be interpreted as a feed-forward layer with weight matrix $\mathbf{I} - \lambda \BB^T \BB$, bias $- \lambda \BB^T\By$, and activation function $\text{prox}_{\alpha \lambda \mathcal{J}}$. Using a proximal mapping as the activation function is not unusual, for example the commonly used ReLU activation is the proximal mapping of the indicator function of $\R_{\ge 0}$.

By viewing the DIP as an explicit optimisation process, the ADP frameworks allows for the application of classical regularisation theory.
For instance, \cite{arndt2022regularization} leverage this connection to establish equivalence between the ADP and classical Ivanov regularisation. 
However, early stopping is a crucial component of the vanilla DIP that was not modelled in the original ADP. 
\cite{arndt2022regularization} resolve this limitation by proposing ADP-$\beta$ to explicitly integrate early stopping in the analysis.

Training the ADP requires solving the bilevel optimisation problem \eqref{eq:adp1}-\eqref{eq:adp2}, which introduces significant computational challenges. 
Traditionally, these have been addressed using two primary approaches: algorithm unrolling, and the implicit function theorem (IFT), which was popularised by deep equilibrium networks \citep{bai2019deep}. In algorithm unrolling the lower-level optimisation problem \eqref{eq:adp2} is truncated to a finite number of iterations. Gradients are then computed by backpropagating through the unrolled steps, which incurs a high memory cost. 
Alternatively, the IFT computes the gradients analytically at the solution of the lower-level problem.
This requires computing the inverse Hessian of the lower-level problem, which is computationally prohibitive for high-dimensional image data. Addressing these computational bottlenecks is currently a highly active area of research. \cite{salehi2025fast} introduced a faster bilevel training framework for the ADP, leveraging MAID \citep{salehi2025adaptively}, an inexact first-order method specifically designed to accelerate bilevel optimisation tasks.

\subsection{Under-parametrised DIP}
\label{sec:deep_decoder}
The deep decoder, introduced by \cite{heckel2018deep}, is an under-parametrised alternative to the usually over-parametrised, autoencoder-style CNNs typically used in DIP.
\textit{Under-parametrisation} in this context refers to the fact that the deep decoder has far fewer parameters than output pixels, effectively mapping a lower-dimensional parameter space to a high-dimensional image space. 
The deep decoder is composed of $d$ fixed upsampling operations, point-wise ($1\times 1$) convolutions, and channel normalisation. 
Formally, the $i$-th layer is given as
\begin{align}
    \Bx^{(i+1)} = \text{cn}(\sigma(\mathbf{U}_i \Bx^{(i)} \mathbf{C}_i)), \quad i=1, \dots, d-1,
\end{align}
where $\mathbf{U}_i$ is a fixed upsampling operator, $\mathbf{C}_i$ are the point-wise convolution weights, $\sigma$ is the activation function, and $\text{cn}$ denotes the channel normalisation. 
In the implementation, \cite{heckel2018deep} use bi-linear interpolation for upsampling and ReLU for activation.
The spatial structure in the deep decoder is given by the upsampling and the network only learns to mix features across channels. This architecture choice makes it inherently regularising, such that it cannot overfit to random noise.
Below, we restate Proposition 1 from \cite{heckel2018deep}, which considers a one layer deep decoder with a single output channel 
\begin{align}
\label{eq:one_layer_dd}
    G((\mathbf{C}_0, c_1)) = \text{ReLU}(\mathbf{U}_0 \Bx^{(0)} \mathbf{C}_0) c_1 \in \R^n ,
\end{align}
with $\mathbf{C}_0 \in \R^{k \times k}$ and $c_1 \in \R^k$ as learnable parameters, and $k$ denoting the number of input channels and $n_0$ the input dimensions.

\begin{proposition}
Consider the one layer deep decoder \eqref{eq:one_layer_dd} with arbitrary upsampling matrix $\mathbf{U}_0 \in \R^{n \times n_0}$ and input $\Bx^{(0)} \in \R^{n_0 \times k}$. Let $\eta \sim \mathcal{N}(0, \sigma^2 \mathbf{I}_n)$. Assume that $k^2  \log(n_0)/n \le 1/32$. Then with probability at least $1-2 n_0^{-k^2}$
\begin{align}
    \min_{(\mathbf{C}_0,c_1)} \| G(\mathbf{C}_0,c_1) - \eta\|_2^2 \ge \| \eta \|_2^2 \left(1 - 20 \frac{k^2 \log(n_0)}{n} \right).
\end{align}
\end{proposition}
This result suggests that the deep decoder can only fit a limited amount of noise, which is determined by the number of parameters relative to the dimension $n$. The authors further provide numerical evidence that a similar inequality also holds for deep decoders with more layers. However, this also limits the peak performance of the deep decoder, making it harder to fit fine textures or complex structures. 

\subsection{Linearised Deep Image Prior}
\label{sec:linearised_dip}
Considering the DIP in the infinite-channel limit, i.e., as the width of the CNN layers tends to infinity, provides a theoreticalframework for analysis, commonly referred to as the linearised DIP. In this over-parametrised regime, when trained using gradient descent with a suitable learning rate, the network's behaviour is governed by the Neural Tangent Kernel (NTK) \citep{jacot2018neural,lee2019wide,arora2019exact}. 
In particular, for a sufficiently over-parametrised CNN, the change in the parameters $\Btheta$ during training becomes negligible \citep{tachella2021neural}. 
Consequently, the behaviour of the network is largely dictated by the spectral properties of the Jacobian at initialisation.
In this regime, we can view the DIP as a linear model using a first order Taylor expansion around the initial parameters $\Btheta^{(0)}$, i.e.,
\begin{align}
    \label{eq:linearised_dip}
   \tilde f_\Btheta(\Bz) = f_{\Btheta^{(0)}}(\Bz) + \BJ (\Btheta - \Btheta^{(0)}),
\end{align}
where the Jacobian $\BJ := \frac{\partial f_\Btheta(\Bz)}{\partial \Btheta}\vert_{\Btheta=\Btheta^{(0)}} \in \R^{n \times p}$ is evaluated at $\Btheta^{(0)}$. 
Substituting the linearisation \eqref{eq:linearised_dip} into the vanilla DIP \eqref{eq:dip_vanilla} transforms the non-linear optimisation problem into a linear least-squares problem 
\begin{align}
        \label{eq:dip_ntk}
    \min_{\Btheta \in \R^p} \frac{1}{2}\| \BA [f_{\Btheta^{(0)}}(\Bz) + \BJ (\Btheta - \Btheta^{(0)})] - \By^\delta \|_2^2,
\end{align}
for the parameter change $\Btheta - \Btheta^{(0)}$. Following \cite{alkhouri2025understanding}, we can now directly write the training dynamics in the space of the network output.

\begin{theorem}\label{th:lin_dip}
Assume that the network output is given by the first order Taylor approximation
\begin{align}\label{eqn:linearised_dip}
   \tilde f_\Btheta(\Bz) = f_{\Btheta^{(0)}}(\Bz) + \BJ(\Btheta - \Btheta^{(0)}),
\end{align}
where $\BJ := \frac{\partial f_\Btheta(\Bz)}{\partial \Btheta}\vert_{\Btheta=\Btheta^{(0)}} \in \R^{n \times p}$ is the Jacobian at initialisation and the DIP is trained using gradient descent with step size $\tau > 0$. Then the training dynamics reduce to 
\begin{align}
    \tilde f_{\Btheta^{(k+1)}}(\Bz) = \tilde f_{\Btheta^{(k)}}(\Bz) - \tau \BJ \BJ^T \BA^T(\BA \tilde f_{\Btheta^{(k)}}(\Bz) - \By^\delta). \label{eq:lin_dip_dynamics}
\end{align}
\end{theorem}

\begin{proof}
Evaluating \eqref{eqn:linearised_dip} for $\Btheta^{(k+1)}$ and $\Btheta^{(k)}$, and subtracting gives
\begin{align}
    \tilde f_{\Btheta^{(k+1)}}(\Bz) &= \tilde f_{\Btheta^{(k)}}(\Bz) + \BJ(\Btheta^{(k+1)} - \Btheta^{(k)}).
\end{align}
Under gradient descent, the change of the parameters is given by 
\begin{align}
    \Btheta^{(k+1)} - \Btheta^{(k)} = - \tau \BJ^T \BA^T(\BA \tilde f_{\Btheta^{(k)}}(\Bz) - \By^\delta),
\end{align}
for step size $\tau >0$. Plugging this in yields
\begin{align}
       \tilde f_{\Btheta^{(k+1)}}(\Bz) = \tilde f_{\Btheta^{(k)}}(\Bz) - \tau \BJ \BJ^T \BA^T(\BA \tilde f_{\Btheta^{(k)}}(\Bz) - \By^\delta).
\end{align}
\end{proof}

\begin{figure}[t]
    \centering
    \begin{subfigure}[b]{1.0\textwidth}
        \centering
        \includegraphics[width=\textwidth]{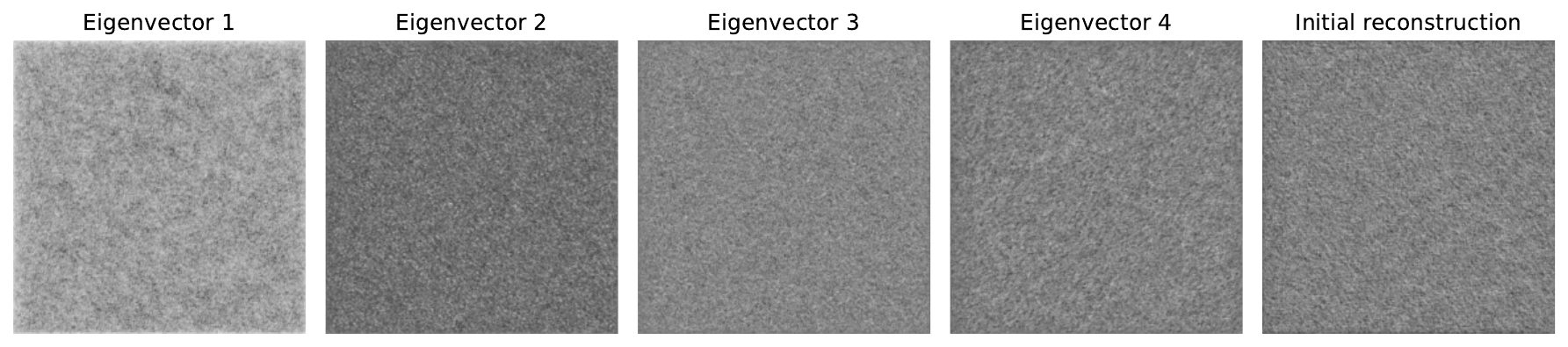}
        \caption{Eigenfunctions of the kernel $\BJ \BJ^T$ with random input $\Bz$ }
        \label{fig:jac_random_inp}
    \end{subfigure}
    
    \begin{subfigure}[b]{1.0\textwidth}
        \centering
        \includegraphics[width=\textwidth]{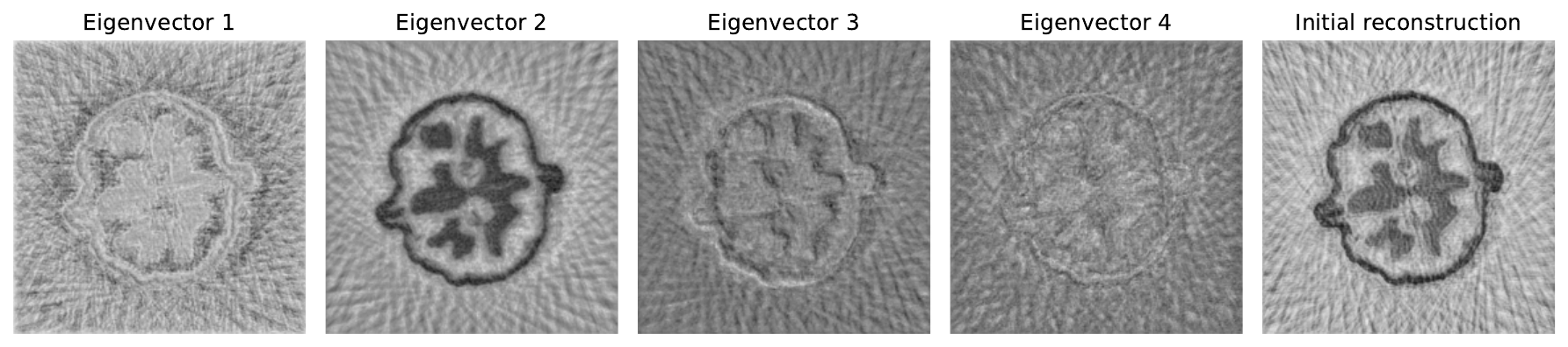}
        \caption{Eigenfunctions of the kernel $\BJ \BJ^T$ with FBP input }
        \label{fig:jac_fbp_inp}
    \end{subfigure}
    \caption{The first four images in each row correspond to the first 4 eigenfunctions of $\BJ \BJ^T$. The last image is the initial output of the (untrained) DIP. In (a) the input to the DIP is random noise, in (b) it is the FBP reconstruction of the $\mu$CT walnut image.}
    \label{fig:jacobian_eigenfunctions}
\end{figure}

The term $\BK := \BJ \BJ^T$ is referred to as the NTK \citep{jacot2018neural}. Assuming we have access to its SVD, we can express the training dynamics \eqref{eq:lin_dip_dynamics} in the eigenbasis as
\begin{align}
    \tilde f_{\Btheta^{(k+1)}}(\Bz) =  \tilde f_{\Btheta^{(k)}}(\Bz) - \tau \sum_{i=1}^n \sigma_i  \langle  \mathbf{g}^{(k)} , \Bu_i  \rangle \Bu_i, \label{eq:lin_dip_dynamics_svd}
\end{align}
where $\mathbf{g}^{(k)} = \BA^T (\BA \tilde f_{\Btheta^{(k)}}(\Bz) - \By^\delta)$ is the update direction and $\Bu_i$ the eigenvector for the eigenvalue $\sigma_i$. The update in \eqref{eq:lin_dip_dynamics_svd} gives rise to the spectral bias of the NTK, meaning that components corresponding to larger eigenvalues $\sigma_i$ are updated more quickly. 
Thus, the DIP prioritises fitting components of the solution aligning with the top eigenspaces. 
The spectral properties of the Jacobian $\BJ$ are both data and model-dependent, as the Jacobian is influenced by both the underlying network architecture and the input $\Bz$. 
To highlight this influence, we show the top four eigenfunctions of a randomly initialised DIP in Figure \ref{fig:jacobian_eigenfunctions}.
We compare the eigenvectors with either a random input $\Bz$ or the FBP $\Bz = \BA^\dagger \By^\delta$ as the initial input to the network, for the same $\mu$CT walnut data used in Figure \ref{fig:vanilla_dip}. 
We see that the eigenvectors for the FBP input more closely resemble the ground truth image. 

This spectral bias has been studied by \cite{liang2025analysis} and it has been shown that the performance of the linearised DIP depends strongly on the relationship between the NTK $\BK$ and the forward operator $\BA$.  
For compressed sensing \cite{liang2025analysis} show that accurate image reconstruction is possible provided the ground truth image $\Bx$ satisfies $\Bx \in \text{Range}(\BK)$, where $\BK$ is singular, and $\BA$ is incoherent with $\BK$. 
Incoherence of $\BA$ and $\BK$ is defined as $\text{Range}(\BK) \cap \text{ker}(\BA) = \emptyset$, 
meaning that when applying the forward operator on the network output the information is preserved. 

For specific two-layer over-parametrised CNNs it can be shown that the linearised model fits natural images faster than noise in the context of denoising \citep[Theorem 2]{heckel2019denoising} and compressed sensing \cite[Theorem 2]{heckel2020compressive}. 
However, Theorem \ref{th:lin_dip} and all the results referred to in this paragraph only hold for training with gradient descent. In practice, standard gradient descent is almost never used for training. Instead, most implementations rely on the adaptive ADAM optimiser, in which case the NTK limit and the first-order Taylor approximation no longer hold \citep{tachella2021neural}. 
In spite of this, the linearised DIP can be used as a tool for uncertainty estimation \citep{antoran2023uncertainty} or Bayesian experimental design \citep{barbano2022bayesian}.

\section{Preventing Overfitting}
\label{sec:overfitting}
Despite the effectiveness of the DIP framework in solving a wide range of imaging inverse problems, it is well known that the vanilla DIP \eqref{eq:dip_vanilla} starts overfitting to measurement noise after a critical number of iterations has been reached, provided the underlying network architecture has sufficient expressive capacity. 
Unlike other learning-based approaches, the DIP does not rely on external training and instead exploits the structural bias of the underlying neural network architecture, as discussed in Section \ref{sec:theoretical_approaches}. 
However, this implicit prior is often insufficient to prevent the network from overfitting to noise, particularly in severely ill-posed settings or under high noise levels. 
It has also been theoretically shown that the DIP can in certain simplified scenarios perfectly fit any signal \citep{van2018compressed}, though this analysis is restricted to single layer overparametrised networks and applications in compressed sensing.

To address this limitation, various strategies have been proposed to improve reconstruction quality and mitigate overfitting. 
These include robust early stopping rules, additional regularisation terms that explicitly penalise undesired solutions or use leverage pre-trained image denoisers, self-guidance and adaptive optimisation methods. In the following, we discuss these approaches and summarise them in Table~\ref{tab:dip_overfitting}. 

\begin{table}[htbp]
    \centering
    \caption{Comparison of reconstruction and regularisation approaches discussed in this section.}
    \label{tab:methods_comparison}
    \begin{tabularx}{\textwidth}{@{} X X X @{}}
        \toprule
        \textbf{Method} & \textbf{Reference} & \textbf{Algorithm} \\
        \midrule
        Vanilla DIP & \cite{ulyanov2018deep} & Algorithm \ref{alg:vanilla_dip} \\ 
        \midrule
        Running Variance Estimator (Early Stopping) & \cite{wang2021early} & Algorithm \ref{alg:early_stopping} \\
        \midrule 
        \multicolumn{3}{l}{\textit{Using additional explicit regulariser}} \\
        \midrule
        DIP + TV (Gradient-based) & \cite{baguer2020computed, liu2019image} & Algorithm \ref{alg:dip+tv} \\
        DIP + TV (HQS) & This work & Algorithm \ref{alg:dip+tv+hqs} \\
        DIP + WTV & \cite{cascarano2021admmtvdip} & ADMM \eqref{eq:wtv_admm1}--\eqref{eq:wtv_admm3}  \\
        Constrained DIP (adaptive WTV) & \cite{cascarano2023constrained} & alternating PGDA \\
        \midrule
        \multicolumn{3}{l}{\textit{Using denoiser (plug-and-play / RED)}} \\ \hline 
        DeepRED & \cite{mataev2019deepred} & Algorithm \ref{alg:dip+red+apg} \\
       
        DIP + Denoiser (ADAM) & This work & Algortihm \ref{alg:dip+denoiser+admm} \\
        \midrule
        \multicolumn{3}{l}{\textit{Using self-guidance}} \\ \midrule 
        Self-Guided DIP & \cite{liang2025analysis} & Algorithm \ref{alg:self_guided_dip} \\
        Autoencoding Sequential DIP (aSeqDIP) & \cite{alkhouri2024image} & Algorithm \ref{alg:aSeqDIP} \\
        \bottomrule
    \end{tabularx}
\label{tab:dip_overfitting}
\end{table}

\subsection{Early Stopping}\label{sec:early_stopping}
Early stopping has been a central component of the vanilla DIP \citep{ulyanov2018deep} and has a long history in both classical optimisation and in deep learning. 
As previously discussed, it is often critical for the vanilla DIP because the objective in \eqref{eq:dip_vanilla} lacks an explicit regularisation term, and the reconstruction is driven purely by the data fidelity.

On a qualitative level, it has been empirically observed that the DIP first captures coarse, low-frequency structures (corresponding to natural image structures), before eventually fitting high-frequency features, including noise.
This behaviour can be partly attributed to the optimisation objective, which solely focuses on minimising the data-fidelity term.
Further insights into this phenomenon can be attained through the lens of the spectral bias and the NTK theory, see Section \ref{sec:theoretical_approaches}.
In particular, according to NTK theory, the early iterations of the optimisation process align with the top eigenspaces of the NTK, while later iterations begin to fit directions associated with small eigenvalues, which often correspond to noise. 
Because of this inherent bias towards principal eigenspaces in the initial optimisation steps, early stopping can be an effective regularisation method.

However, selecting an effective early stopping point remains a major practical challenge. 
In inverse problems, one typically only has access to noisy measurements $\By^\delta$ and no ground truth data is available to directly monitor the reconstruction quality. 
As a result, early stopping rules are usually designed using heuristics or through visual inspection of the reconstruction (though the latter is not always practically viable), making the approach highly problem-dependent.
A non-exhaustive list of stopping rules for DIP includes methods that use effective degrees of freedom to monitor the optimisation \citep{jo2021rethinking}, a learned auto-encoder that monitors image quality \citep{li2021self}, using quality metrics that do not require a reference image \citep{benfenati2025early}, and the running variance approach described below.

\paragraph*{Running Variance Estimator}
\cite{wang2021early} investigate strategies to detect the near-peak performance point, as measured by PSNR, by minimising a rolling estimate of the reconstruction variance across optimisation steps.
This is used as a proxy for the reconstruction error.
In particular, they define the running variance 
\begin{align}
    \text{VAR}(k) \coloneqq \frac{1}{W} \sum_{w=0}^{W-1} \left\| \Bx^{(k+w)} - \frac{1}{W} \sum_{i=0}^{W-1} \Bx^{(k+i)} \right\|^2,
\end{align}
where $\Bx^{(k)}$ is the DIP reconstruction at iteration $k$ and $W \in \N$ is the window size. The intuition is that $\text{VAR}(k)$ follows a U-shaped curve: the DIP quickly learns the general structure of the image at the start of training, resulting in a large VAR. At the valley of the curve, the reconstructions $\{\Bx^{(k)}\}$ cluster around a stable solution. As training continues and the DIP begins fitting measurement noise, high-frequency components are amplified, causing the running variance to increase again.
We present the early stopping algorithm for a general metric $g(k)$ in Algorithm~\ref{alg:early_stopping}.
Experimental results show that this criterion often yields stable and robust performance for image denoising, but it can be severely suboptimal for medical imaging, such as CT. 

\begin{algorithm}[H]
\caption{Early stopping criterion (adapted from \cite{wang2021early})}
\label{alg:early_stopping}
\begin{algorithmic}[1]
\Require Metric $g(k)$ evaluated at iteration $k$, patience parameter $\mathfrak{p}$, relative decrease threshold $\delta$
\State $g_{\text{min}} \gets \infty$ \Comment{Best metric value seen so far}
\State $k \gets 0$ \Comment{Current iteration}
\State $k_{\text{min}} \gets \infty$ \Comment{Iteration with best metric value}
\While{$k \leq k_{\text{min}} + \mathfrak{p}$}
    \If{$g(k) < \delta \cdot g_{\text{min}}$}
        \State $g_{\text{min}} \gets g(k)$
        \State $k_{\text{min}} \gets k$
    \EndIf
    \State $k \gets k + 1$
\EndWhile
\State \Return $k_{\text{min}}$
\end{algorithmic}
\end{algorithm}

\subsection{Total Variation Regularisation}
DIP implicitly regularises the inverse problem by favouring reconstructions that can be generated by a convolutional neural network. 
While effective in many cases, this implicit prior can be insufficient for severely ill-posed inverse problems. 
To address this limitation and further guide the reconstruction towards desirable features, recent approaches incorporate an explicit regulariser $\mathcal{J}$ to the DIP \eqref{eq:dip_vanilla}, giving the penalised objective
\begin{align}
    \label{eq:dip+reg}
    \min_{\Btheta \in \R^p} \frac{1}{2}\| \BA f_\Btheta(\Bz) - \By^\delta \|_2^2 + \alpha \mathcal{J}(f_\Btheta(\Bz)).
\end{align}
This approach has been successfully applied in both CT \citep{baguer2020computed} and image restoration tasks \citep{liu2019image}. In these studies, the regulariser is chosen as the (isotropic) total variation (TV) \citep{rudin1992nonlinear}, ${\rm TV}(\Bx)=\sum_{i=1}^n \|(\BD\Bu)_i\|_2$, where  $\BD\Bu=(\BD_x\Bu, \BD_y\Bu)\in\R^{n\times 2}$, is the discretisation of the gradient operator $\nabla$, giving rise to the \textit{DIP+TV} formulation 
\begin{align}
    \label{eq:dip+tv}
    \min_{\Btheta \in \R^p} \frac{1}{2}\| \BA f_\Btheta(\Bz) - \By^\delta \|_2^2 + \alpha {\rm TV}(f_\Btheta(\Bz)).
\end{align}
Although the TV regulariser is non-smooth, the penalised objective \eqref{eq:dip+reg} is in practice still minimised using gradient-based optimisers and automatic differentiation, see Algorithm \ref{alg:dip+tv}.\footnote{Note that such violations of non-smoothness are rather common in deep learning, as most neural network architectures use non-smooth activation functions (e.g., ReLU or Leaky ReLU), such that \eqref{eq:dip+reg} is non-smooth even for a smooth regulariser.} 

\begin{algorithm}[t]
\caption{Deep Image Prior + TV \citep{baguer2020computed}}
\label{alg:dip+tv}
\begin{algorithmic}[1]
\Require Forward operator $\mathbf{A}$, observation $\mathbf{y}^\delta$, maximum number of iteration $N_\text{max}$, model input $\Bz$, model architecture $f_\Btheta$, regularisation parameter $\alpha$
\For{$k=1$ to $N_\text{max}$}
    \State $\tilde{\Bx} =  f_\Btheta(\Bz)$
    \State $L(\Btheta) = \frac{1}{2} \| \BA \tilde{\Bx} - \By^\delta \|_2^2  + \alpha \text{TV}(\tilde{\Bx}) $
    \State Take gradient step w.r.t.~$\Btheta$ using $\nabla L(\Btheta)$
\EndFor
\State \textbf{return} $f_\Btheta(\Bz)$
\end{algorithmic}
\end{algorithm}

\paragraph*{DIP+TV via Half Quadratic Splitting} 
We now explore an alternative to the gradient-based optimisation of the DIP+TV. In particular, we employ splitting methods to treat the smooth and non-smooth components separately. 
To our knowledge, this approach has not been previously applied in the context of DIP regularisation.
By introducing an auxiliary variable $\Bu \in \R^n$, with a constraint $\Bu = f_\Btheta(\Bz)$, the problem \eqref{eq:dip+reg} can be reformulated as
\begin{align}\label{eq:dip+tv_hqs}
    \min_{\Btheta \in \R^p, \Bu \in \R^n} \frac{1}{2}\| \BA f_\Btheta(\Bz) - \By^\delta \|_2^2 + \alpha \mathcal{J}(\Bu) + \frac{\mu}{2} \| f_\Btheta(\Bz) - \Bu \|_2^2.
\end{align}
Formulations of this type are commonly referred to as half-quadratic splitting (HQS) methods. 
The main advantage of \eqref{eq:dip+tv_hqs} lies in the decoupling of network parameters $\Btheta$ from the regulariser, which is applied only to the auxiliary variable~$\Bu$. This augmented formulation is solved via alternating minimisation, allowing $\Btheta$ to be optimised independently of the regularisation term~$\mathcal{J}$, as
\begin{align}
    \Btheta^{(k+1)} &= \argmin_{\Btheta \in \R^p} \frac{1}{2}\| \BA f_\Btheta(\Bz) - \By^\delta \|_2^2 + \frac{\mu_k}{2} \| f_\Btheta(\Bz) - \Bu^{(k)} \|_2^2 \label{eq:altern1} \\ 
    \Bu^{(k+1)} &= \argmin_{\Bu \in \R^n}  \mathcal{J}(\Bu) + \frac{\mu_k}{2\alpha} \| f_{\Btheta^{(k+1)}}(\Bz) - \Bu \|_2^2, \label{eq:altern2}
\end{align}
for some initial $\Bu^{(0)}, \Btheta^{(0)}$ and an increasing schedule for $\mu_k$, $k=0, \dots, K-1$. Pseudo code for this splitting approach is given in Algorithm \ref{alg:dip+tv+hqs}.

For a convex regularisation functional $\mathcal{J}$, \eqref{eq:altern2} can be expressed using the proximal mapping
\begin{align}
    \Bu^{(k+1)} &= \text{prox}_{\alpha/\mu_k \mathcal{J}}(f_{\Btheta^{(k+1)}}(\Bz)). 
\end{align}
Since the proximal operator is known in closed form only for a select few functions, evaluating the above equation typically requires computing an approximate solution of another minimisation problem in each iteration in order to compute $\text{prox}_{\alpha/\mu_k\mathcal{J}}$.
If $N_{\rm sub}$ gradient updates are used to solve the original problem, allocating $K/N_{\rm sub}$ updates to each subproblem in the alternating scheme maintains the same computational cost concerning neural network evaluations. 
While computing the proximal mapping can be computationally expensive, solving \eqref{eq:dip+reg} also incurs the cost of evaluating $\nabla \mathcal{J}$ at each gradient step. 
Further, instead of solving the proximal mapping using an iterative method, one can also approximate the solution, e.g., \cite{chandler2025closed} propose a closed-form approximation for TV. Further, the computation cost can be reduced by applying ideas such as ProxSkip \citep{papoutsellis2026proxskip}, where the proximal mapping is not evaluated at every iteration.

\begin{algorithm}[t]
\caption{Deep Image Prior + TV (HQS)}
\label{alg:dip+tv+hqs}
\begin{algorithmic}[1]
\Require Forward operator $\mathbf{A}$, observation $\mathbf{y}^\delta$, splitting strength $\lambda > 0$, number of outer iterations $K$, number of inner iterations $N_\text{sub}$, model input $\Bz$, sequence of regularisation parameters $\{\alpha_k\}_{k=1}^K$
\State Initialise $\mathbf{u}^{(0)}=\mathbf{0}$ 
\For{$k=1$ to $K$}  
    \State $\mu_k = \lambda / \alpha_k$
    \For{$i=1$ to $N_\text{sub}$}
    \State $\tilde{\Bx} =  f_\Btheta(\Bz)$ 
    \State $L(\Btheta) = \frac{1}{2} \| \BA \tilde{\Bx} - \By^\delta \|_2^2 + \frac{\mu_k}{2} \| \tilde{\Bx}- \mathbf{u}^{(k-1)} \|_2^2$
    \State Take gradient step w.r.t.~$\Btheta$ using $\nabla_\Btheta L(\Btheta)$
    \EndFor
\State $\mathbf{u}^{(k)} = \text{prox}_{\alpha_k \text{TV}}(f_\Btheta(\Bz))$ 
\EndFor
\State \textbf{return} $f_\Btheta(\Bz)$
\end{algorithmic}
\end{algorithm}

\paragraph*{Weighted Total Variation and DIP via Splitting}
In \cite{cascarano2021admmtvdip} the authors consider the weighted TV (WTV) regulariser $\reg(\Bu)=\sum_{i=1}^n \alpha_i \|(\BD\Bu)_i\|_2$,
where $n$ is the image size, i.e., $\Bx\in\R^n$, $\BD\Bu=(\BD_x\Bu, \BD_y\Bu)\in\R^{n\times 2}$, is the discretisation of the gradient operator $\nabla$, and $(\alpha_i)_{i=1}^n$ are pixel-wise weights.
The corresponding DIP optimisation problem
\begin{equation}
\min_{\Btheta \in \R^p} \frac{1}{2} \| \BA f_{\Btheta}(\Bz) - \By \|_2^2 + \sum_{i=1}^{n} \alpha_i \| (\BD f_{\Btheta}(\Bz))_i \|_2,
\end{equation}
can then also be solved via splitting methods. 
Introducing the constraint $\BD f_{\Btheta}(\Bz) = \Bu$
the problem can be solved\footnote{See also \url{https://github.com/sedaboni/ADMM-DIPTV} for implementation details} with alternating direction method of multipliers (ADMM) by adding a Lagrangian $\Bla \in \R^{n \times 2}$ as 
\begin{equation}
\begin{aligned}
\frac12\| \BA f_{\Btheta}(\Bz) - \By\|_2^2 
+ \sum_{i=1}^{n} \alpha_i\| u_i \|_2  + \frac{\gamma}{2} \| \BD f_{\Btheta}(\Bz) - \Bu \|_2^2 
+ \langle \Bla, \BD f_{\Btheta}(\Bz) - \Bu \rangle,
\end{aligned}
\end{equation}
where $\gamma>0$.
This can now be solved with ADMM iterations
\begin{align}
\label{eq:wtv_admm1}
\Btheta^{(k+1)} &\in \argmin_{\Btheta} \; \frac 12 \| \BA f_{\Btheta}(\Bz) - \By \|_2^2 
+ \frac{\gamma}{2} \left\| \BD f_{\Btheta}(\Bz) - \Bu^{(k)} + \Bla^{(k)}/\gamma \right\|_2^2 \\
\Bu^{(k+1)} &= \argmin_{\Bu} \; \sum_{i=1}^{n} \alpha_i \| u_i \|_2 
+ \frac{\gamma}{2} \left\| \Bu - \left( \BD f_{\Btheta^{(k+1)}}(\Bz) + \Bla^{(k)}/\gamma \right) \right\|_2^2 \\
\Bla^{(k+1)} &= \Bla^{(k)} + \gamma \left( \BD f_{\Btheta^{(k+1)}}(\Bz) - \Bu^{(k+1)} \right). \label{eq:wtv_admm3}
\end{align}
Note that the iterations above have a fundamentally different goal and behaviour to the iterations in Algorithm \ref{alg:dip+tv+hqs}. Namely, whereas the constraint in \eqref{eq:dip+tv_hqs} pairs image sized object, above the constraint pairs image gradient-sized objects. Whereas the advantage of the former is that regulariser and the weights are decoupled (which thus avoids backpropagating through the non-smooth TV term), it comes at a cost of having to compute the proximal operator.
On the other hand, ADMM iterations require backpropagating through a (translated) TV term, but the proximal term (the update for $\Bu$), is fully separable in $n$ variables (each belonging to $\R^2$), and has a closed form solution.

Setting all weights $\alpha_i$ constant recovers the standard TV regulariser, but the weights $(\alpha_i)_{i=1}^n$ can also be adaptive, by for example setting 
$\alpha_i=\frac 1n \frac{\|\BA f_\Btheta(\Bz)-\By\|^2}{\|(\BD f_\Btheta(\Bz))_i \|}$.
The splitting approach discussed in this section  can be extended to more general cases, where the regulariser is written as $\reg(\Bx)=\sum_{i=1}^n \alpha_i \reg_i(\BB_i f_\theta(\Bx))$, with adaptive weights updated analogously \citep{cascarano2023constrained}.

The latter work also considers using the Morozov discrepancy principle to select the regularisation weights, through the following constrained formulation of the regulariser DIP
\begin{equation}
\label{eq:constrained_dip}
\begin{split}
&\argmin_{\Btheta \in \R^p} \reg(f_\Btheta(\Bz)), \\ &\text{with } f_\Btheta(\Bz)\in S_\sigma:=\{f_\Btheta(\Bz) ~ \vert ~ \|\BA f_\Btheta(\Bz) -\By\|_2^2\leq \tau \sigma^2 m\},
\end{split}
\end{equation}
where $\tau>0$, $\sigma>0$ is the standard deviation of the measurement noise and $m$ is the dimension of the measurements. 
This is then solved via an alternating proximal gradient descent-ascent (PGDA) method\footnote{See \url{https://github.com/pcascarano/cDIP-RED-DIP-WTV} for implementation details}. 

\subsection{DIP with Learned Denoisers}
Instead of a hand-crafted regulariser $\reg$ we can also use an off-the-shelf denoiser in the 
variational DIP formulation \eqref{eq:dip+reg}, inspired of the success of plug-and-play denoising approaches \citep{zhang2021plug}. 

\textbf{Regularisation by denoising (RED)}
To our knowledge, the first such example for DIP was the RED regulariser \citep{romano2017little}, defined as 
\begin{align*}
    \reg (\Bx) = \frac 12 \Bx^T (\Bx-{\rm D}_\sigma(\Bx)),
\end{align*}
where ${\rm D}_\sigma$ is a denoiser and $\sigma$ the denoising strength. 
Under some conditions
on ${\rm D}_\sigma$, it can be shown that the resulting regulariser is convex and differentiable, and that its gradient is given by $\nabla \reg(\Bx) =\Bx-{\rm D}_\sigma(\Bx)$, see Section 5 in \cite{romano2017little} for details.
Thus, unlike for TV, using RED does not require computing the gradient of a non-differentiable regulariser in training.
Instead, computing the gradient requires evaluating the denoiser. 
When used for DIP, the resulting method is called DeepRED \citep{mataev2019deepred}\footnote{See \url{https://github.com/GaryMataev/DeepRED} for  for implementation details}.
It is important to note however that most denoisers used in practice (e.g., BM3D) do not satisfy the conditions that ensure $\nabla \reg(\Bx) =\Bx-{\rm D}_\sigma(\Bx)$ holds.
In particular, it can be shown that if the Jacobian of $\rm D_\sigma$ is not symmetric then a regulariser such that $\nabla \reg(\Bx) =\Bx-{\rm D}_\sigma(\Bx)$ holds does not exist \cite[Theorem 1]{reehorst2018regularization}.
Nevertheless, RED based methods still often show strong empirical performance, and some further modifications can be used to ensure their convergence.
In particular, \cite{reehorst2018regularization} reframe RED through the lens of denoising score-matching, which aims to approximate the gradient of the log-prior (the score) rather than the prior density itself
Leveraging this connection. the authors derive alternative, provably convergent RED-based algorithms.
In Algorithm \ref{alg:dip+red+apg} we provide an application of one such algorithm, using Nesterov accelrated proximal gradient, for a DIP based optimisation solver. 

\begin{algorithm}[t]
\caption{Deep Image Prior + RED (APG)}
\label{alg:dip+red+apg}
\begin{algorithmic}[1]
\Require Forward operator $\mathbf{A}$, observation $\mathbf{y}^\delta$, number of outer iterations $K$, number of inner iterations $N_\text{sub}$, model input $\Bz$, splitting strength $\lambda>0$, $t_0=1$ and $L>1$, denoiser $\rm D_\sigma$
\State Initialise $\mathbf{u}^{(0)}=\mathbf{0}$ 
\For{$k=1$ to $K$}  
    \For{$i=1$ to $N_\text{sub}$}
    \State $\tilde{\Bx} =  f_\Btheta(\Bz)$ 
    \State $L(\Btheta) = \frac{1}{2} \| \BA \tilde{\Bx} - \By^\delta \|_2^2 + \frac{\lambda L}{2}\| \tilde{\Bx}- \mathbf{u}^{(k-1)} \|_2^2$
    \State Take gradient step w.r.t.~$\Btheta$ using $\nabla_\Btheta L(\Btheta)$
    \EndFor
\State $t_k=\frac{1+\sqrt{1+4t_{k-1}^2}}{2}$    
\State $\mathbf{z}^{(k)} = \mathbf{x}^{(k)}+\frac{t_{k-1}-1}{t_k} (\mathbf{x}^{(k)}-\mathbf{x}^{(k-1)})$
\State $\mathbf{u}^{(k)} = \frac{1}{L}\rm D_\sigma(\mathbf{z}^{(k)}) - \frac{1-L}{L}\mathbf{z}^{(k)}$
\EndFor
\State \textbf{return} $f_\Btheta(\Bz)$
\end{algorithmic}
\end{algorithm}

\textbf{Splitting based approaches}
An alternative for incorporating  denoisers is to employ splitting methods, such as HQS or ADMM.
For instance, motivated by the new Algorithm \ref{alg:dip+tv+hqs}, we apply HQS to derive the following minimisation iterations
\begin{align}
    \Btheta^{(k+1)} &= \argmin_{\Btheta \in \R^p} \frac{1}{2}\| \BA f_\Btheta(\Bz) - \By^\delta \|_2^2 + \frac{\mu_k}{2} \| f_\Btheta(\Bz) - \Bu^{(k)} \|_2^2  \\ 
    \Bu^{(k+1)} &= \argmin_{\Bu \in \R^n}  \reg(\Bu) + \frac{\mu_k}{2\alpha} \| f_{\Btheta^{(k+1)}}(\Bz) - \Bu \|_2^2.
\end{align}

As discussed earlier, for a convex regulariser $\mathcal J$ the last step is equivalent to applying the proximal operator of $\reg$, which can be replaced with a denoiser
\begin{align*}
\Bu^{(k+1)} &= {\rm D}_{\alpha/\mu_k}(f_{\Btheta^{(k+1)}}(\Bz)).
\end{align*}
This formulation does not impose special conditions on the denoiser (such as those needed for RED), but requires tuning the schedule of weights $\mu_k$.
This can be avoided by using ADMM instead of HQS, which introduces the constraint and adds a Lagrangian, see Algorithm \ref{alg:dip+denoiser+admm}.

\begin{algorithm}[t]
\caption{Deep Image Prior + Denoiser (ADMM)}
\label{alg:dip+denoiser+admm}
\begin{algorithmic}[1]
\Require Forward operator $\mathbf{A}$, observation $\mathbf{y}^\delta$, ADMM coupling penalty $\beta > 0$, number of outer iterations $K$, number of inner iterations $N_\text{sub}$, model input $\Bz$, denoiser $D_\sigma$
\State Initialise $\mathbf{u}^{(0)}=\mathbf{0}$, $\Vmu^{(0)}=\mathbf{0}$
\For{$k=1$ to $K$}  
    \For{$i=1$ to $N_\text{sub}$}
    \State $\tilde{\Bx} =  f_\Btheta(\Bz)$ 
    \State $L(\Btheta) = \frac{1}{2} \| \BA \tilde{\Bx} - \By^\delta \|_2^2 + \frac{\beta}{2}\| \tilde{\Bx}- \mathbf{u}^{(k-1)}+\Vmu^{(k-1)} \|_2^2$
    \State Take gradient step w.r.t.~$\Btheta$ using $\nabla_\Btheta L(\Btheta)$
    \EndFor
\State $\Bx^{(k)} =  f_\Btheta(\Bz)$
\State $\mathbf{u}^{(k)} = \rm D_{\sigma}(\mathbf{x}^{(k)}+\Vmu^{(k-1)})$
\State $\Vmu^{(k)} = \Vmu^{(k-1)}+ (\mathbf{x}^{(k)}-\mathbf{u}^{(k)})$
\EndFor
\State \textbf{return} $f_\Btheta(\Bz)$
\end{algorithmic}
\end{algorithm}

\subsection{Self-Guidance Methods}
Empirical studies have shown that the choice of the network input can have a significant impact on the reconstruction quality.
Instead of using random noise as input, some approaches use a coarse reconstruction, e.g., the filtered backprojection for CT \citep{barbano2022educated}, or incorporate available reference images, e.g., structurally similar images \citep{zhao2020reference}. 
These informed input images provide additional information to guide the DIP, which aligns with observations in Figure \ref{fig:jacobian_eigenfunctions} showing that NTK eigenfunctions are indeed sensitive to the choice of the network input. 
However, the obvious limitation of such approaches is the reliance on the suitable reference image. 

Self-guided DIP \citep{liang2025analysis} overcomes this limitation by setting up a joint optimisation problem in which the network parameters and the network input are optimised at the same time. 
In particular, the self-guided DIP is trained using the following loss function  
\begin{equation}
\label{eq:self_guided_dip}
    (\hat{\Btheta}, \hat{\Bz}) \in \argmin_{\Btheta \in \R^p, \Bz \in \R^n} \frac{1}{2}\| \BA \E_\eta[f_\Btheta(\Bz + \eta)] - \By^\delta \|_2^2 + \frac{\alpha}{2}\|  \E_\eta[f_\Btheta(\Bz + \eta)] - \Bz \|_2^2.
\end{equation}
The first term is the usual data fidelity term, enforcing consistency between the reconstructed image and the observed data. Unlike the original DIP formulation, this term is not evaluated at a single network output $f_\Btheta(\Bz)$, but at the expectation $\E_\eta[f_\Btheta(\Bz + \eta)]$, where $\eta$ denotes random perturbations of the input. The second term is a self-guided denoising (or consistency) term. It penalises deviations between the expected network output and the input $\Bz$. Overall, the objective couples reconstruction and denoising: the network learns parameters that explain the measurements, while simultaneously adapting its input toward a stable, denoised representation. \cite{liang2025analysis} observe that the method is robust against the choice of noise distribution for $\eta$, and obtain similar results for both Gaussian and uniform noise. 
Their experiments use uniform noise $\eta \sim U(0,m)$ with a maximum value of $m = \max(|\Bz|)/2$, corresponding to a rather high noise level compared to traditional denoising methods. 
In practice, the noise level might need some tuning for a given inverse problem at hand.
The final reconstruction is taken as the denoised mean, either as  $\hat{\Bx} = \E_\eta[f_{\hat{\Btheta}}(\hat{\Bz} + \eta)]$ or using an exponentially moving average over training iterations\footnote{\url{https://github.com/sjames40/Self-Guided-DIP}},  
\begin{align}
    \Bx^{(k)} = \gamma \Bx^{(k-1)} + (1-\gamma) \E_\eta[f_{\Btheta^{(k)}}(\Bz^{(k)} + \eta)],
\end{align}
where $(\Btheta^{(k)}, \Bz^{(k)})$ are the parameters at the $k$-th training step. The full algorithm for the self-guided DIP is given in Algorithm \ref{alg:self_guided_dip}. 

\begin{algorithm}[t]
\caption{Self-Guided DIP \citep{liang2025analysis}}
\label{alg:self_guided_dip}
\begin{algorithmic}[1]
\Require Forward operator $\mathbf{A}$, observation $\mathbf{y}^\delta$, regularisation strength $\lambda > 0$, weighting $\alpha \in [0,1)$, $N_\eta$ number of noise samples, noise distribution $P_\eta$, maximum number of iteration $N_\text{max}$
\State $\Bx^{(0)} = \mathbf{0}$
\For{$k=1$ to $N_\text{max}$}
    \State $\tilde{\Bx} = \frac{1}{N_\eta} \sum_{i=1}^{N_\eta} f_\Btheta(\Bz + \eta_i), \quad \eta_i \sim P_\eta$
    \State $L(\Btheta, \Bz) = \frac{1}{2} \| \BA \tilde{\Bx} - \By^\delta \|_2^2 + \frac{\lambda}{2} \| \tilde{\Bx}- \Bz \|_2^2$
    \State Take gradient step w.r.t.~$(\Btheta,\Bz)$ using $\nabla L(\Btheta, \Bz)$
    \State $\Bx^{(k)} = \alpha \Bx^{(k-1)} + (1-\alpha) \tilde{\Bx}$
\EndFor
\State \textbf{return} $\mathbf{x}^{(N_\text{max})}$
\end{algorithmic}
\end{algorithm}

The self-guided DIP has a higher computational cost compared to the vanilla DIP due to two contributing factors.
First, the gradients are computed with respect to not only the network parameters but also the network input. 
Second, each training step requires computing the expectation of the DIP's output with respect to multiple noise initialisations. 

Building on top of the self-guided DIP, \cite{alkhouri2024image} recently proposed the \emph{autoencoding sequential DIP} (aSeqDIP), which considers a sequential optimisation problem instead of the joint optimisation problem in \eqref{eq:self_guided_dip}. 
The aSeqDIP is trained by iteratively solving the following optimisation problem
\begin{align}
\label{eq:aSeqDIP_1}
    \Btheta^{(k+1)} &= \argmin_{\Btheta \in \R^p}  \frac{1}{2}\| \BA f_\Btheta(\Bz^{(k)}) - \By^\delta \|_2^2 + \frac{\lambda}{2}\|  f_\Btheta(\Bz^{(k)}) - \Bz^{(k)} \|_2^2 \\ \label{eq:aSeqDIP_2}
    \Bz^{(k+1)} &= f_{\Btheta^{(k+1)}}(\Bz^{(k)}),
\end{align}
for $k=1, \dots, K$, and with $\Bz^{(0)} = \BA^T \By^\delta$ and $\Btheta^{(0)}$ as the initial network parameters. 
Subproblem \eqref{eq:aSeqDIP_1} needs to be solved using an additional optimisation scheme, and to speed-up the convergence the parameters are initialised using the solution of the previous problem $\Btheta^{(k-1)}$. 
However, the computational cost of the aSeqDIP can be similar to training a vanilla DIP. 
Assuming we use $N_\text{sub}$ gradient steps for every subproblem in \eqref{eq:aSeqDIP_1}, aSeqDIP has a similar computational cost as a DIP trained for $K N_\text{sub}$ iterations. 
The only computational difference is the computation of $K$ forward passes of the network in \eqref{eq:aSeqDIP_2} after solving each subproblem. 
The full algorithm is provided in Algorithm \ref{alg:aSeqDIP}.

\begin{algorithm}[t]
\caption{aSeqDIP \citep{alkhouri2024image}}
\label{alg:aSeqDIP}
\begin{algorithmic}[1]
\Require Forward operator $\mathbf{A}$, observation $\mathbf{y}^\delta$, regularisation strength $\lambda > 0$, number of outer iter. $K$, number of inner iter. $N_\text{sub}$, initialisation $\Bz^{(0)}$
\For{$k=1$ to $K$}
    \For{$i=1$ to $N_\text{sub}$}
    \State $\tilde{\Bx} =  f_\Btheta(\Bz^{(k-1)})$ 
    \State $L(\Btheta) = \frac{1}{2} \| \BA \tilde{\Bx} - \By^\delta \|_2^2 + \frac{\lambda}{2} \| \tilde{\Bx}- \Bz^{(k-1)} \|_2^2$
    \State Take gradient step w.r.t.~$\Btheta$ using $\nabla_\Btheta L(\Btheta)$
    \EndFor
    \State $\Bz^{(k)} =  f_\Btheta(\Bz^{(k-1)})$ 
\EndFor
\State \textbf{return} $\Bz_{K}$
\end{algorithmic}
\end{algorithm}

\section{Computational Considerations}\label{sec:computational_cost}

The DIP is typically deployed by training a CNN from scratch for each reconstruction, which is often computationally prohibitive.
This is in stark contrast to the standard supervised paradigm in deep learning, where once a network has been trained it can be efficiently deployed at inference time.
To address this limitation, several computational strategies have been developed to accelerate the DIP.
The primary computational bottleneck arises from differentiating either through the network or through the forward operator, and in particularly unfavourable cases, both operations may incur substantial overhead.
In this section, we focus on three strategies to reduce the computational time.
First, the optimisation can be accelerated by warm-starting with appropriately selected network weights instead of random initialisations, thereby decreasing the number of required gradient updates.
Second, the optimisation can be constrained to a suitable subspace, effectively reducing the dimensionality of the parameter space and enabling the use of second-order optimisation methods. 
Lastly, instead of using the gradient of the data-fidelity term with respect to the full forward operator, stochastic gradient updates over subsampled operators can be employed.

\subsection{Warm-starting the Optimisation} 
\label{sec:edip}
As discussed in Section~\ref{sec:algo_design_choices}, initialisation of the DIP weights plays an important role in improving and accelerating optimisation.
In the standard framework, the network is initialised randomly and trained from scratch for every new measurement. 
However, when ground truth data is available, or when informative data can be simulated, they can be leveraged to identify a suitable set of initial weights.
This is achieved in \cite{barbano2022educated} through a two-stage process.
The first stage consists of supervised pretraining of the network on a dataset $\mathcal{D} = \{(\mathbf{x}^n, \mathbf{y}^n)\}_{n=1}^N$ of images $\Bx$ and corresponding measurements $\By$.
This is done by optimising
\begin{equation}\label{eq:supervised_loss}
    \hat{\Btheta} \in \argmin_{\Btheta \in \R^p}N^{-1}\mkern-25mu\sum_{(\Bx^{n},\By^{n})\in \mathcal{D}}\mkern-18mu\| f_{\Btheta}(\BA^\dagger\By^{n}) - \Bx^{n}\|_2^2.
\end{equation}
In the second stage the network parameters are fine-tuned in order to adapt to the target reconstruction task, by solving either \eqref{eq:dip_vanilla} or \eqref{eq:dip+reg}.
The results show that the pretrained DIP exhibits strong input robustness.
To explain this effect, \cite{barbano2022educated} propose a spectral analysis by linearising the pre-trained neural network, see also Section~\ref{sec:linearised_dip}.
Specifically, the spectrum of the Jacobian of the linearised network around the pretrained parameters shows that pretraining induces a sparsification of the basis while imposing a relevance shift.
In other words, it induces a hierarchical structure that manifests as a shift towards the parameters of the decoder.
This shift is coherent with the observation that most of the heavy-lifting of representing the target image is actually done by the decoder.

Note that warm-starting is only advisable when the
pretraining dataset closely resembles the data to be reconstructed, as substantial distributional differences may lead to poor feature transfer and, consequently, degraded reconstruction performance. However, assessing such similarity is more an art than a science.
Warm-starting is a form of test-time training (TTT) \citep{sun2020test}, as the network parameters are adapted on the fly for a new task.
\cite{darestani2022test} apply the TTT framework to image reconstruction, and show that TTT can improve the performance on unseen data.

\subsection{Subspace DIP}
In a manner complementary to network pretraining, the computational costs can be mitigated by constraining the optimisation to a sparse linear subspace of the parameters \citep{barbano2024image}.
The low-dimensionality of the subspace reduces the tendency to fit noise (alleviating the need for early stopping) and enables the use of stable second order optimisation methods, e.g., natural gradient descent \citep{amari1998natural} or limited-memory Broyden–Fletcher–Goldfarb–Shanno \citep{liu1989limited}.

The methodology boils down to two steps.
The first step consists of identifying a subspace of parameters that is low-dimensional and easy to work with, but contains a rich enough set of parameters to fit $\By$.
This is achieved by leveraging the pre-training trajectory, i.e., the sequence of network parameters $\Btheta^{(i)} \in \R^p$ recorded over the course of the supervised optimisation in \eqref{eq:supervised_loss}, to construct basis vectors. In particular, $K$ parameter vectors are sampled at uniformly spaced checkpoints and then stacked into a matrix $\Theta \in \mathbb{R}^{p \times K}$. 
A low-dimensional parameter subspace is then identified using singular vectors from an approximate SVD, $\Theta \approx \mathbf{USV}^T,$
where top-$k$ left singular vectors $\mathbf{U}\in\mathbb{R}^{p \times k}$ are retained, and $k \ll K$ is the dimensionality of the chosen subspace.
The orthonormal basis $\mathbf{U}$ can be further sparsified by computing leverage scores \citep{drineas2012fast}, associated with each DIP parameter as
\[
\rho_{i} = \sum_{j=1}^{k} [\mathbf{U}]^{2}_{ij},\;i=1,\dots,p,
\]
and retaining the basis vector entries corresponding to the $d_{\tau} \ll p$ largest leverage scores.
This can be achieved by applying a mask $\mathbf{M} \in \{0,1\}^{p \times p}$ satisfying  $[\mathbf{M}]_{ii} = \mathbf{1}_{\operatorname{top}_{d_\tau}}(i)$ 
where $\operatorname{top}_{d_{\tau}}$ is the set of indices of $d_\tau$ largest leverage scores $\rho_1, \rho_2, \ldots, \rho_{p}$.
The sparse basis $\mathbf{MU}$ contains at most $d_{\tau} \cdot {k}$ non-zero entries.
The objective \eqref{eq:var_obj} is then rewritten in terms of the sparse basis coefficient $\boldsymbol{c}\in \mathbb{R}^{{k}}$,
\begin{equation}\label{eq:sub-dip}
    \argmin_{\boldsymbol{c} \in \R^k} L(\gamma(\boldsymbol{c})),\;\rm{where}\;\gamma(\boldsymbol{c}):= \Btheta + \mathbf{MU}\boldsymbol{c},
\end{equation}
and $L(\gamma) = \frac{1}{2} \| \BA f_\gamma(\Bz) - \By^\delta \|_2^2$.
This restricts the DIP parameters $\Btheta \in \R^p$ from pre-training with \eqref{eq:supervised_loss} to change only along the sparse subspace $\mathbf{MU}$.
In particular, the number of optimisation variables are drastically reduced as the number of sparse basis coefficients is way smaller than the number of network parameters. 
Note that \eqref{eq:sub-dip} differs from traditional deep learning loss functions in that the local curvature matrix can be computed and stored accurately and efficiently, without resorting to restrictive approximation of its structure.

\subsection{Stochastic Gradient Descent}
A main computational bottleneck of solving variational problems such as \eqref{eq:var_obj} using gradient-based methods is the cost of the evaluation of the forward operator.
This is generally attributed to the cost of matrix-vector products, which scales linearly with both the number of input and output dimensions.
For example, the gradient of the least-squares data fidelity term, i.e.,
\begin{align*}
    \nabla_\Bx \frac 12 \|\BA\Bx-\By^\delta\|_2^2 = \BA^T(\BA \Bx - \By^\delta),
\end{align*}
requires one evaluation of both the forward operator and its adjoint, which in large scale tomography pipelines (e.g., in CT or PET), can incur a high per-iteration computational cost.
This is different from typical issues in deep learning optimisation, which is concerned with the computational costs with respect to the network parameters \citep{bottou2018optimization}.

This gave rise to methods that aim to reduce the per-iteration computational cost by using only a part of the measurement data in each iteration. We refer to \cite{ehrhardt2025guide} for a recent review of these methods. 
In inverse imaging problems this is done by partitioning the forward operator $\BA$ into $N_b$ components and defining functions $\ell_i$ such that $\ell(\BA\Bx;\By^\delta)= \sum_{i=1}^{N_b}\ell_i(\BA_i\Bx;\By_i^\delta)$.
In the simplest, stochastic gradient descent case, the full gradient $\nabla_\Bx \ell(\BA\Bx;\By^\delta)$ can now be estimated by sampling a random index $i$ and computing a gradient only of $\ell(\BA_i\Bx;\By_i^\delta)$ \citep{jin2018regularizing}.

This methodology has shown a strong performance in practice and has led to remarkable speed-ups in medical image reconstruction regimes, see \cite{chambolle2018stochastic, twyman2022investigation, tang2020practicality} for a few examples. It is important to emphasise that there are many algorithmic variations, improvements, and modifications that target different components of the reconstruction pipeline. 

It is straightforward to extend this to the DIP framework. Instead of computing the gradient of $\frac 12 \|\BA f_\Btheta(\Bz)-\By\|_2^2$, with respect to $\Btheta$, which would involve applying the full forward operator $\BA$ on the current network output,  the gradient of $\frac{N_b}{2}\|\BA_i f_\Btheta(\Bz)-\By_i\|_2^2$ can instead be computed.
However, due to the additional computational overhead involved in backpropagating through the network with respect to $\Btheta$, which is not the case for non-deep learning based approaches, the computational benefits might be limited.

\section{Numerical Experiments}
\label{sec:numerical_experiments}
In this section, we first compare several approaches aimed at mitigating overfitting problem in the DIP. 
We then analyse the computational cost of the vanilla DIP relative to the methods introduced in Section \ref{sec:computational_cost}. 
Experiments are conducted on real-world $\mu$CT data (see Section \ref{sec:dataset} for details), using the U-Net architecture described in Section \ref{sec:network_architecture}. 
To ensure a fair comparison, all methods use the same random seed, yielding identical initial network parameters and random inputs (where applicable). 
While the performance of the DIP is highly problem-dependent, our results are intended to highlight trends and provide practical insights. In particular, we compare 
\begin{itemize}
    \item Vanilla DIP (see Alg. \ref{alg:vanilla_dip})
    \item Vanilla DIP with FBP input 
    \item DIP+TV using the joint optimisation (see Alg. \ref{alg:dip+tv})
    \item DIP+TV using the splitting approach (see Alg. \ref{alg:dip+tv+hqs})
    \item Self-guidance DIP (see Alg. \ref{alg:self_guided_dip})
    \item aSeqDIP (see Alg. \ref{alg:aSeqDIP})
    \item DIP+RED (with APG, see Alg. \ref{alg:dip+red+apg})
    \item DIP+Denoiser (with HQS and with ADMM, see Alg. \ref{alg:dip+denoiser+admm})
\end{itemize}
All hyperparameters, e.g., regularisation strengths, step size or number of inner iterations, are chosen to maximise the peak-signal-to-noise-ratio (PSNR). 
The methods that employ a pretrained denoiser use DRUNet \citep{zhang2021plug}, obtained from the deepinv library\footnote{\url{https://deepinv.github.io/deepinv/}}, and exponential averaging of the output. The code for all experiments is publicly available.\footnote{\url{https://github.com/alexdenker/deep-image-prior-handbook}} 

\begin{table}[t]
\caption{Performance comparison of different DIP-based methods. The TV baseline achieves a PSNR of $26.88$dB.}
\centering
\begin{tabular}{lccccc}
\toprule
 & \multicolumn{2}{c}{\textbf{Peak Result}} 
 & \multicolumn{2}{c}{\textbf{Early Stopping}} 
 & \textbf{Final} \\ 
\cmidrule(lr){2-3}\cmidrule(lr){4-5}
\textbf{Method} & PSNR & Iter. & PSNR & Iter. & PSNR \\
\midrule
Vanilla DIP & $28.58$ & $1282$ & $26.10$ & $2189$ & $17.55$ \\
Vanilla DIP (FBP input) & $28.61$ & $1293$ & $28.04$ & $917$ & $17.61$ \\
DIP+TV (joint, Alg.~\ref{alg:dip+tv}) & $28.40$ & $1595$  & $27.52$  & $2785$ & $26.32$ \\
DIP+TV (splitting, Alg.~\ref{alg:dip+tv+hqs}) & $28.66$ & $1405$ & $28.27$ & $1405$ & $26.59$ \\
Self-Guidance DIP & $28.83$ & $9858$ & $26.11$ & $2554$ & $28.82$ \\
aSeqDIP & $28.71$ & $6880$ & $28.15$ & $5700$ & $18.66$ \\
DIP+RED (APG, Alg,~\ref{alg:dip+red+apg}) & $28.70$ & $2640$ & $27.37$ & $4440$ & $24.04$ \\
DIP+Denoiser (HQS) & $30.04$ & $5980$ & $29.33$ & $4680$ & $27.18$ \\
DIP+Denoiser (ADMM, Alg. \ref{alg:dip+denoiser+admm}) & $30.98$ & $9460$ & $30.97$ & $10000$ & $30.97$ \\
\bottomrule
\end{tabular}
\label{tab:dip_results}
\end{table}

\subsection{\texorpdfstring{$\mu$}{Micro-}CT Walnut Dataset}
\label{sec:dataset}
For the numerical evaluation, we consider a real high-resolution $\mu$CT dataset. 
Specifically, we reconstruct a $(501\,\mathrm{px})^2$ slice from a highly sparse cone-beam scan of a walnut. 
The data are acquired using 60 projection angles and 128 detector pixels, resulting in a total of $m = 7680$ measurements.
Reconstructions are compared against a ground truth obtained via classical reconstruction from original full data acquired at three source positions, with 1200 angles and 768 detector pixels.
The data is taken from the public walnut dataset published by \cite{der2019cone}, which comprises 42 specimens. 
This task requires capturing both coarse structures and fine image details, making it a suitable proxy for industrial $\mu$CT applications.
We use measurements of ``Walnut~1'' acquired from source position (orbit)~2.
For the 2D reconstruction, we select the slice offset by $+3$\,px from the mid-plane.
Further details can be found in the supplementary material of \cite{barbano2022educated}, and the implementation of the geometry used here is available at \href{https://github.com/educating-dip/educated_deep_image_prior}{educated-dip.github.io}.

\subsection{Network Architecture}
\label{sec:network_architecture}
We use the U-Net architecture \citep{ronneberger2015u} following the implementation in \cite{barbano2022educated}, replacing bilinear upsampling with nearest neighbour upsampling\footnote{In the current PyTorch version (version 2.3.1) bilinear upsampling is non deterministic on GPU, even with a fixed seed.}, ensuring consistent and reproducible results. 
The network consists of six scales (encoder/decoder depth), each with $128$ feature channels. 
The input resolution is $501 \times 501$, and the encoder progressively downsamples the feature maps to resolutions of $251\times 251$, $126 \times 126$, $63 \times 63$, $32 \times 32$, and $16 \times 16$. 
Skip connections are added from the encoder to the decoder at the two lowest resolution levels. 
All convolutional layers use a Leaky ReLU activation function. 
In total, the network has roughly $3$M parameters.

\captionsetup[subfigure]{skip=2pt}
\begin{figure}[htbp]
    \centering

    \makebox[0.25\textwidth][c]{\small\textbf{Best Reconstruction}}%
    \hspace{0.02\textwidth}%
    \makebox[0.25\textwidth][c]{\small\textbf{Early stopping}}%
    \hspace{0.02\textwidth}%
    \makebox[0.25\textwidth][c]{\small\textbf{Final Reconstruction}}\\[4pt]
    
    \begin{subfigure}{0.25\textwidth}
        \includegraphics[width=\linewidth]{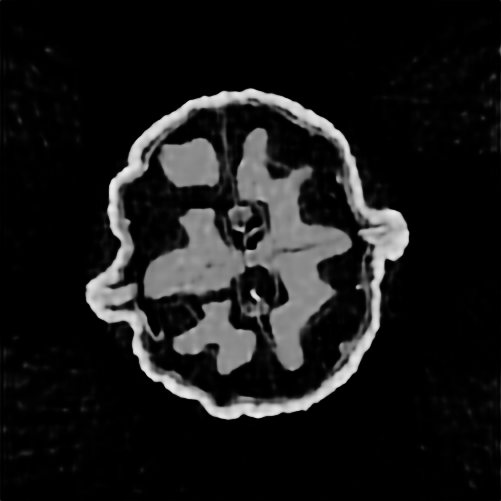}
        \caption*{}
    \end{subfigure}
    \begin{subfigure}{0.25\textwidth}
        \includegraphics[width=\linewidth]{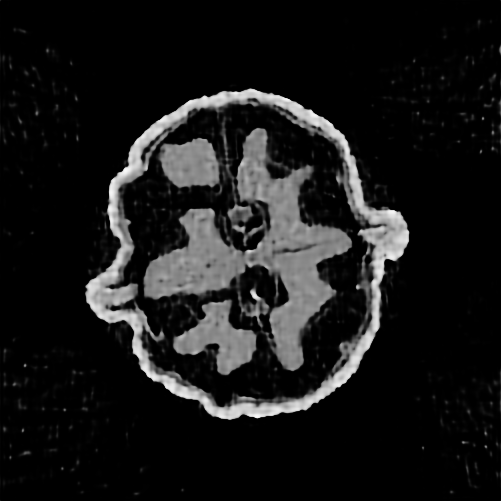}
        \caption*{\textbf{Vanilla DIP}}
    \end{subfigure}
    \begin{subfigure}{0.25\textwidth}
        \includegraphics[width=\linewidth]{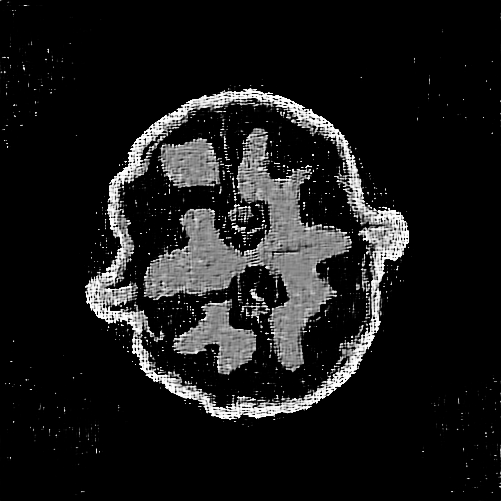}
        \caption*{}
    \end{subfigure}

    \vspace{-10pt}
    \rule{0.8\textwidth}{0.4pt}
    \vspace{2pt}

    \begin{subfigure}{0.25\textwidth}
        \includegraphics[width=\linewidth]{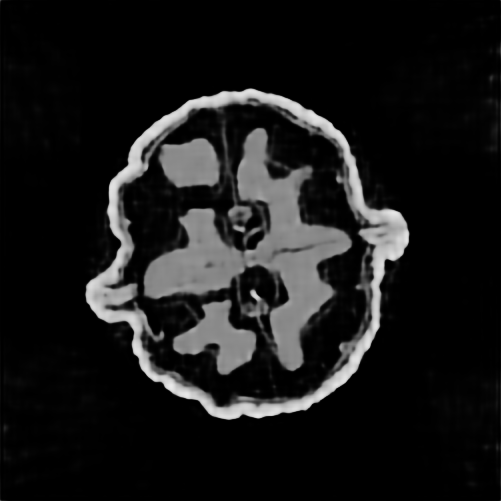}
        \caption*{}
    \end{subfigure}
    \begin{subfigure}{0.25\textwidth}
        \includegraphics[width=\linewidth]{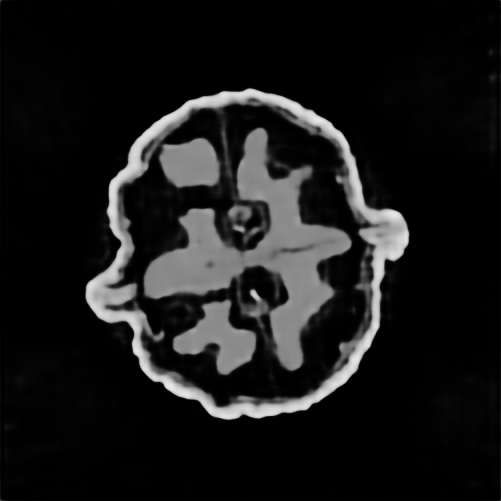}
        \caption*{\textbf{Vanilla DIP (FBP input)}}
    \end{subfigure}
    \begin{subfigure}{0.25\textwidth}
        \includegraphics[width=\linewidth]{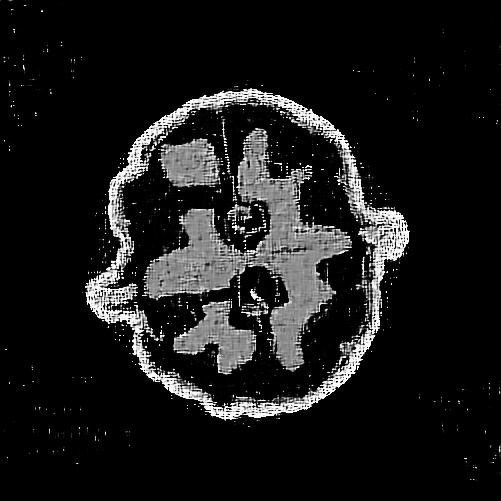}
        \caption*{}
    \end{subfigure}

    \vspace{-10pt}
    \rule{0.8\textwidth}{0.4pt}
    \vspace{2pt}

    \begin{subfigure}{0.25\textwidth}
        \includegraphics[width=\linewidth]{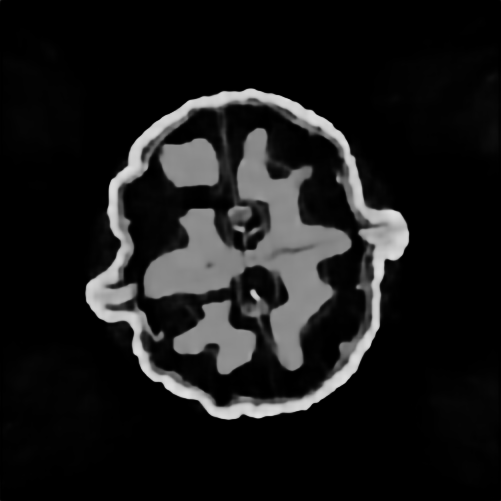}
        \caption*{}
    \end{subfigure}
    \begin{subfigure}{0.25\textwidth}
        \includegraphics[width=\linewidth]{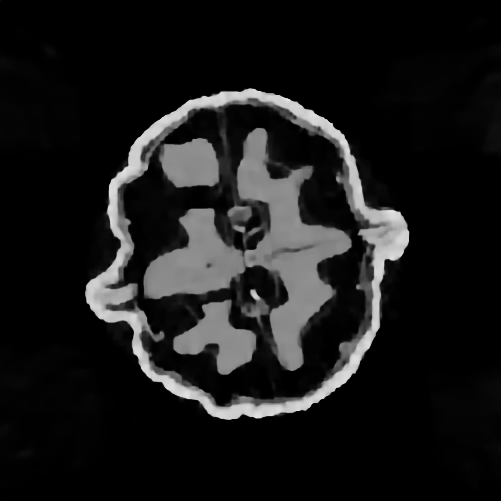}
        \caption*{\textbf{DIP+TV (joint, Alg.\ref{alg:dip+tv})}}
    \end{subfigure}
    \begin{subfigure}{0.25\textwidth}
        \includegraphics[width=\linewidth]{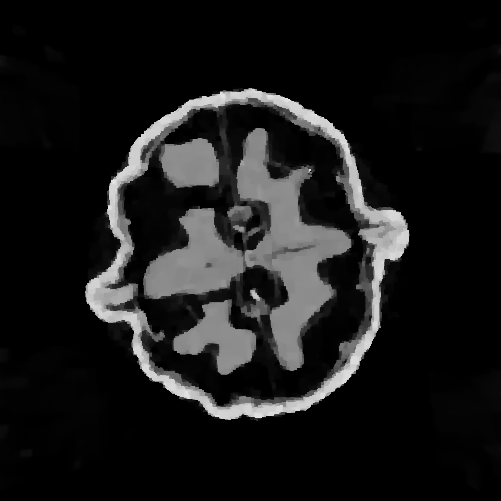}
        \caption*{}
    \end{subfigure}

    \vspace{-10pt}
    \rule{0.8\textwidth}{0.4pt}
    \vspace{2pt}

    \begin{subfigure}{0.25\textwidth}
        \includegraphics[width=\linewidth]{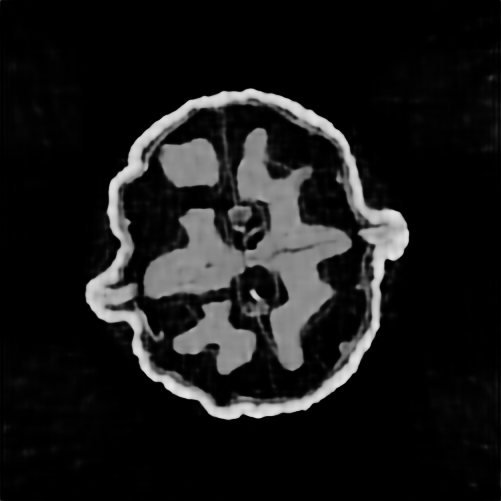}
        \caption*{}
    \end{subfigure}
    \begin{subfigure}{0.25\textwidth}
        \includegraphics[width=\linewidth]{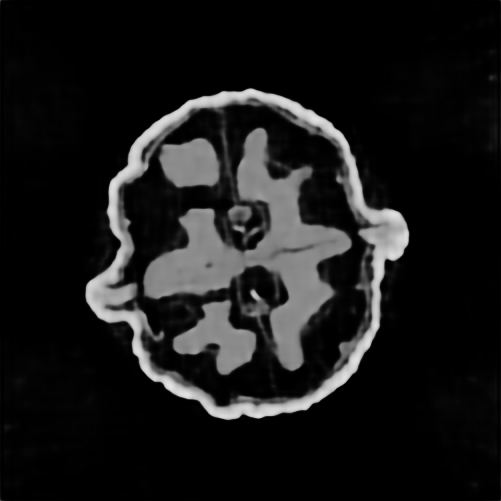}
    \caption*{\textbf{DIP+TV (splitting, Alg.~\ref{alg:dip+tv+hqs})}}
    \end{subfigure}
    \begin{subfigure}{0.25\textwidth}
        \includegraphics[width=\linewidth]{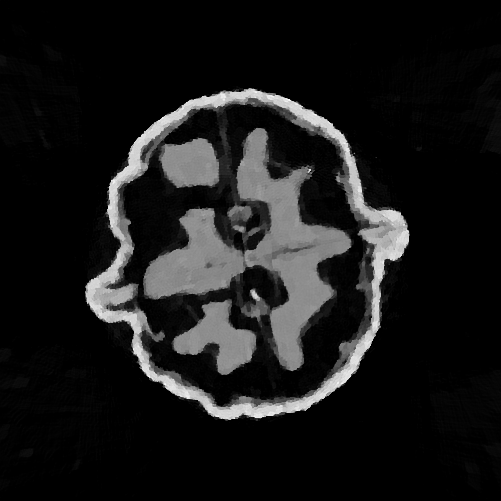}
        \caption*{}
    \end{subfigure}

    \caption{$\mu$CT walnut reconstructions for the different DIP methods. Left: Best reconstruction. Middle: Early stopping reconstruction. Right: Final reconstruction.}
    \label{fig:all_recos}
\end{figure}

\begin{figure}[t]
    \centering
    \begin{subfigure}[t]{0.48\linewidth}
        \centering
        \includegraphics[width=\linewidth]{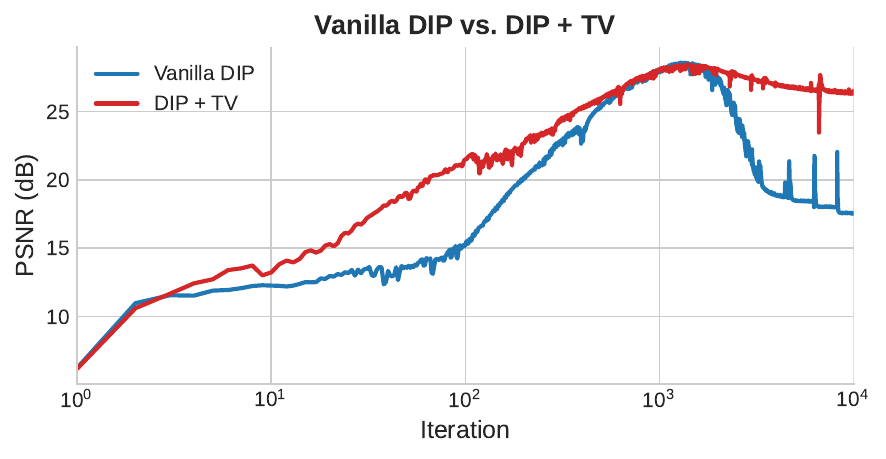}
        \caption{Vanilla DIP vs. DIP+TV (joint).}
        \label{fig:vanilla_vs_tv}
    \end{subfigure}
    \hfill
    \begin{subfigure}[t]{0.48\linewidth}
        \centering
        \includegraphics[width=\linewidth]{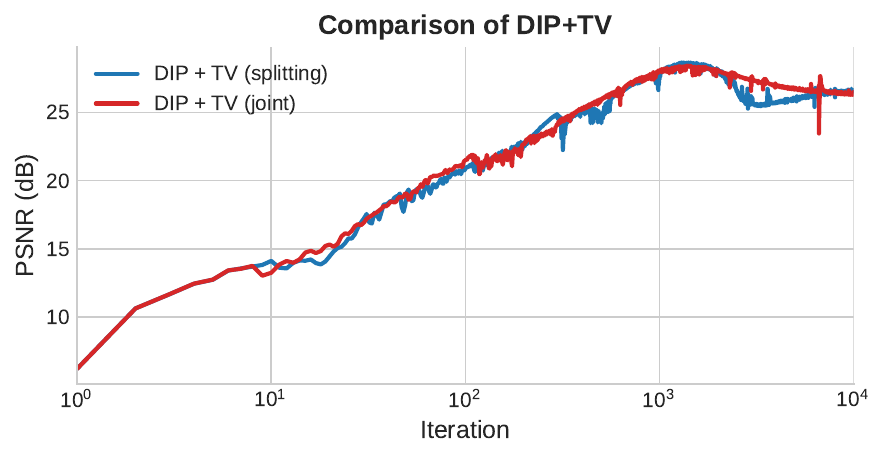}
        \caption{Joint vs. splitting (HQS).}
        \label{fig:joint_vs_splitting}
    \end{subfigure}
    \caption{Evolution of PSNR during optimisation for all considered DIP variants.}
    \label{fig:dip_all_comparison}
\end{figure}

\subsection{Results}

We show quantitative results in Table~\ref{tab:dip_results}.
We evaluate three PSNR metrics: the peak PSNR, defined as the maximum PSNR achieved during optimisation; the early stopping PSNR, obtained using the running variance estimator described in Section~\ref{sec:early_stopping}; and the final PSNR, measured after $\num{10000}$ iterations. 
As a classical baseline, we also compute a TV-regularised reconstruction, which reaches a PSNR of $26.87$ dB. 
The TV regularisation parameter was tuned to directly maximise PSNR, providing an optimistic upper bound for traditional TV-based methods.

Figure~\ref{fig:all_recos} and Figure~\ref{fig:all_recos_v2} presents qualitative results, showing the reconstruction with the highest PSNR, the early-stopped reconstruction, and the final reconstruction.
As expected, we observe a degradation of the image quality for vanilla DIP. 
Adding additional regularisation into the objective, as for the DIP+TV, can mitigate this issue:
the best reconstruction and the reconstruction after \num{10000} iterations show only minor visual differences. 
Similarly, both self-guidance DIP and aSeqDIP do not rely as strongly on early stopping and show fewer image artifacts at the final iteration. 
Lastly, methods that employ a deep denoiser instead of TV achieve the highest peak PSNR, in particular the splitting-based approaches. 
This observation is similar to plug-and-play approaches, where deep denoiser often achieve better results \citep{zhang2017learning,zhang2021plug}. 
Moreover, increasing the coupling parameter that enforces the constraint reduces the maximal achievable PSNR, but also mitigates overfitting.

In Figure~\ref{fig:vanilla_vs_tv} we further investigate this behaviour by comparing the evolution of the PSNR for vanilla DIP and DIP+TV. 
The results show that their peak PSNR is nearly identical, which is expected since they use the same U-Net architecture and initialisation.
Note that the peak PSNR for all methods is higher than the optimistic TV-regularised reconstruction. 
However, their optimisation dynamics differ substantially.
The vanilla DIP exhibits a sharp peak, i.e., a narrow window of high PSNR, whereas DIP+TV produces a broader, more stable peak, simplifying the selection of a suitable stopping point.
Moreover, after $\num{10000}$ iterations, DIP+TV yields a higher final PSNR, suggesting reduced overfitting relative to vanilla DIP.
Note that the early stopping rule was developed for the vanilla DIP and it seems to not adapt well to the aSeqDIP. 
However, we notice that early stopping was not necessary for the aSeqDIP and we observe less overfitting, due to the additional regularisation in \eqref{eq:aSeqDIP_1}. 

\captionsetup[subfigure]{skip=2pt}
\begin{figure}[t]
    \centering

    \makebox[0.25\textwidth][c]{\small\textbf{Best Reconstruction}}%
    \hspace{0.02\textwidth}%
    \makebox[0.25\textwidth][c]{\small\textbf{Early stopping}}%
    \hspace{0.02\textwidth}%
    \makebox[0.25\textwidth][c]{\small\textbf{Final Reconstruction}}\\[4pt]

    \begin{subfigure}{0.25\textwidth}
        \includegraphics[width=\linewidth]{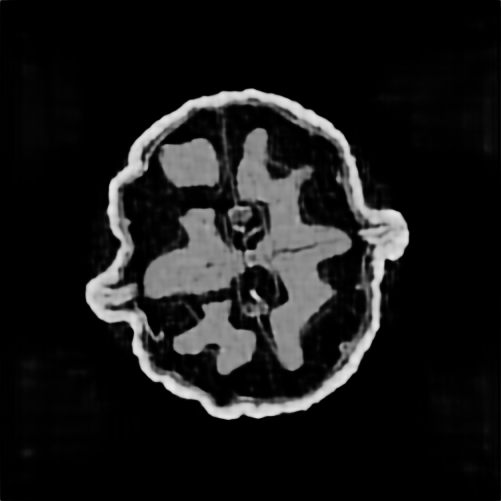}
        \caption*{}
    \end{subfigure}
    \begin{subfigure}{0.25\textwidth}
        \includegraphics[width=\linewidth]{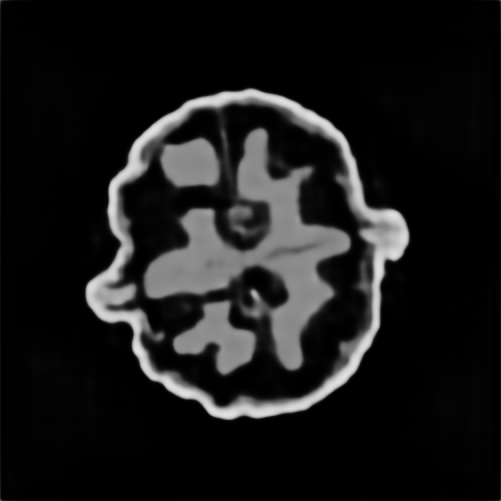}
        \caption*{\textbf{Self-Guidance DIP}}
    \end{subfigure}
    \begin{subfigure}{0.25\textwidth}
        \includegraphics[width=\linewidth]{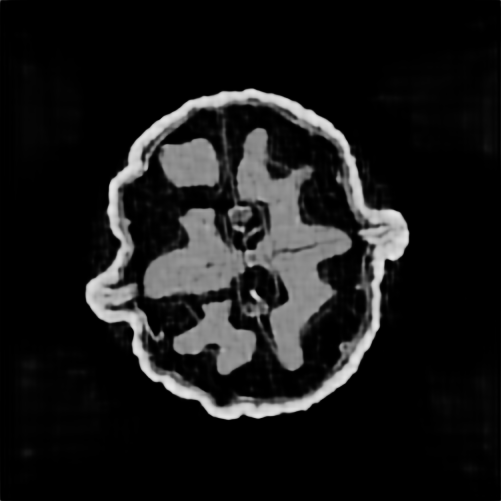}
        \caption*{}
    \end{subfigure}

    \vspace{-10pt}
    \rule{0.8\textwidth}{0.4pt}
    \vspace{2pt}

    \begin{subfigure}{0.25\textwidth}
        \includegraphics[width=\linewidth]{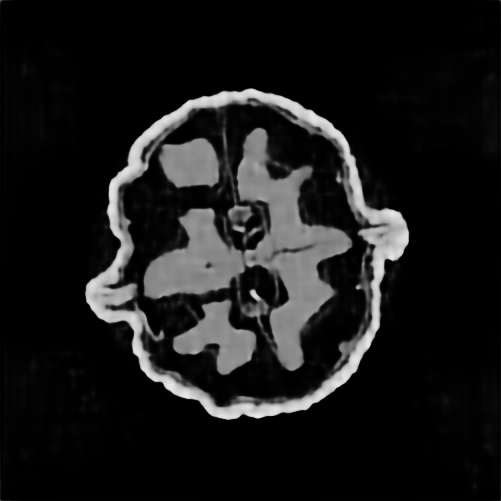}
        \caption*{}
    \end{subfigure}
    \begin{subfigure}{0.25\textwidth}
        \includegraphics[width=\linewidth]{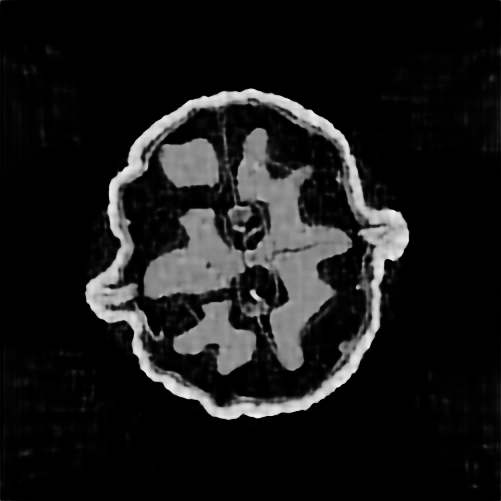}
        \caption*{\textbf{aSeqDIP}}
    \end{subfigure}
    \begin{subfigure}{0.25\textwidth}
        \includegraphics[width=\linewidth]{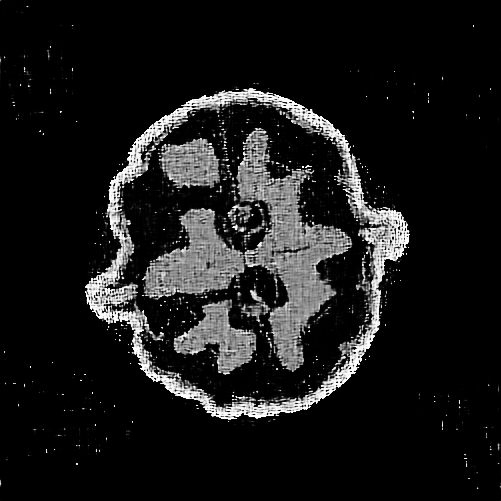}
        \caption*{}
    \end{subfigure}

    \caption{$\mu$CT walnut reconstructions for the different DIP methods. Left: Best reconstruction. Middle: Early stopping reconstruction. Right: Final reconstruction.}
    \label{fig:all_recos_v2}
\end{figure}

\captionsetup[subfigure]{skip=2pt}
\begin{figure}[t]
    \centering

    \makebox[0.25\textwidth][c]{\small\textbf{Best Reconstruction}}%
    \hspace{0.02\textwidth}%
    \makebox[0.25\textwidth][c]{\small\textbf{Early stopping}}%
    \hspace{0.02\textwidth}%
    \makebox[0.25\textwidth][c]{\small\textbf{Final Reconstruction}}\\[4pt]

    \begin{subfigure}{0.25\textwidth}
        \includegraphics[width=\linewidth]{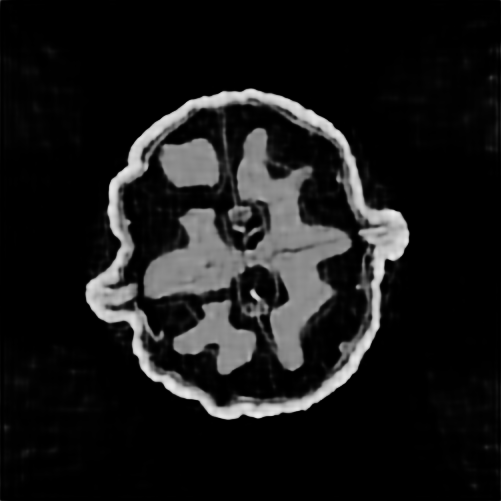}
        \caption*{}
    \end{subfigure}
    \begin{subfigure}{0.25\textwidth}
        \includegraphics[width=\linewidth]{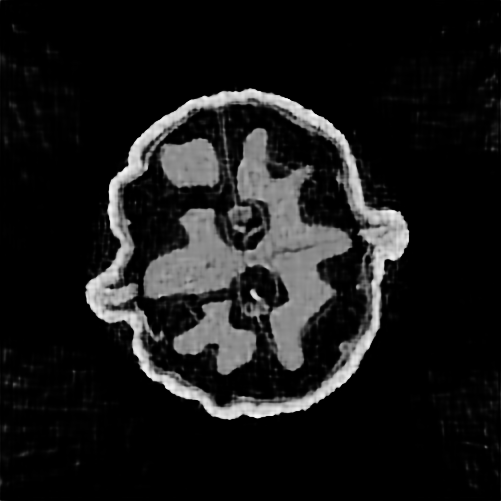}
        \caption*{\textbf{DIP+RED}}
    \end{subfigure}
    \begin{subfigure}{0.25\textwidth}
        \includegraphics[width=\linewidth]{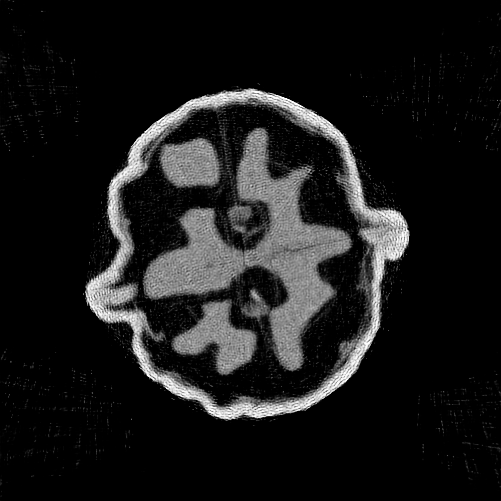}
        \caption*{}
    \end{subfigure}

    \vspace{-10pt}
    \rule{0.8\textwidth}{0.4pt}
    \vspace{2pt}

    \begin{subfigure}{0.25\textwidth}
        \includegraphics[width=\linewidth]{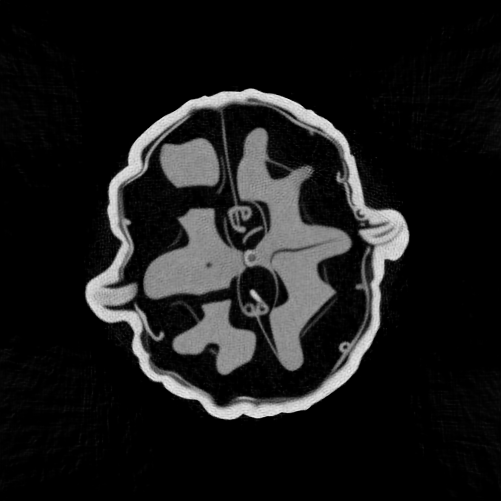}
        \caption*{}
    \end{subfigure}
    \begin{subfigure}{0.25\textwidth}
        \includegraphics[width=\linewidth]{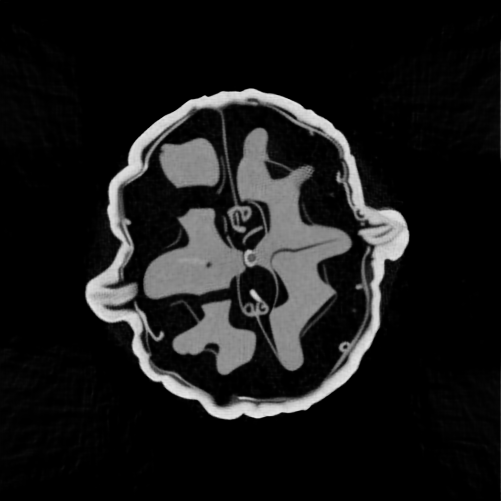}
        \caption*{\textbf{DIP + Denoiser (HQS)}}
    \end{subfigure}
    \begin{subfigure}{0.25\textwidth}
        \includegraphics[width=\linewidth]{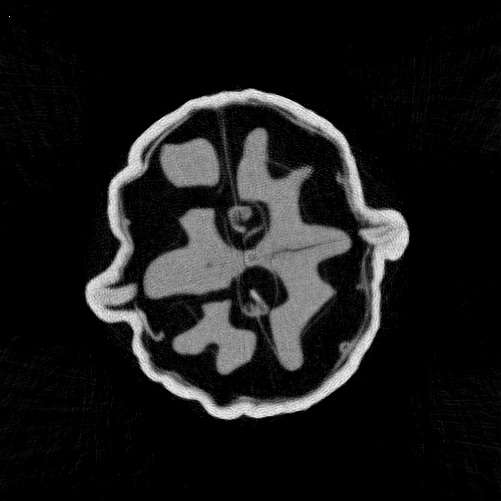}
        \caption*{}
    \end{subfigure}

    \vspace{-10pt}
    \rule{0.8\textwidth}{0.4pt}
    \vspace{2pt}

    \begin{subfigure}{0.25\textwidth}
        \includegraphics[width=\linewidth]{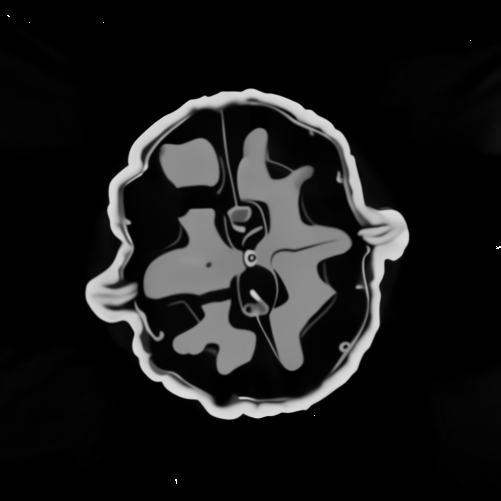}
        \caption*{}
    \end{subfigure}
    \begin{subfigure}{0.25\textwidth}
        \includegraphics[width=\linewidth]{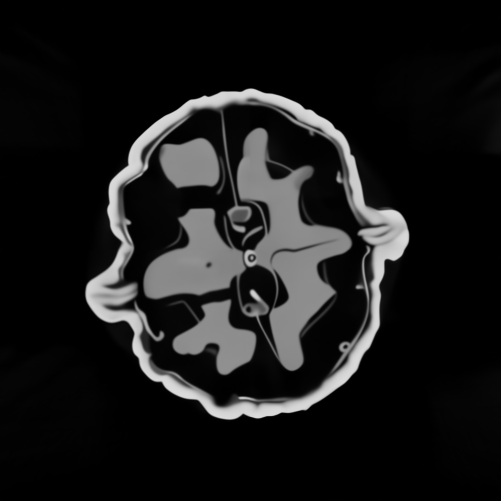}
        \caption*{\textbf{DIP+Denoiser (ADMM)}}
    \end{subfigure}
    \begin{subfigure}{0.25\textwidth}
        \includegraphics[width=\linewidth]{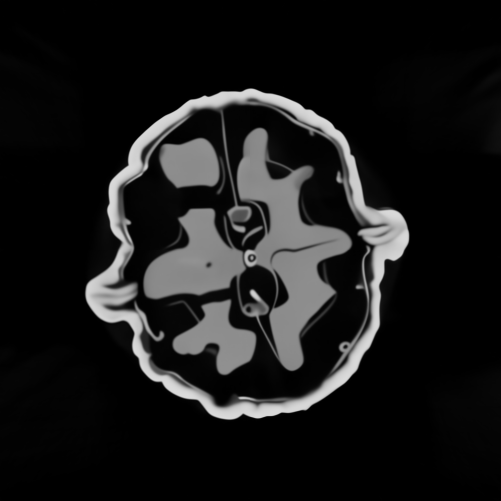}
        \caption*{}
    \end{subfigure}    

    \caption{$\mu$CT walnut reconstructions for the different DIP methods. Left: Best reconstruction. Middle: Early stopping reconstruction. Right: Final reconstruction.}
    \label{fig:all_recos_v3}
\end{figure}


In Figure~\ref{fig:joint_vs_splitting}, we compare the PSNR trajectories of the two TV-regularised DIP variants.
The difference is minor, with the splitting approach reaching a slightly higher peak PSNR.
However, since the splitting method incurs higher computational cost, as it is necessary to evaluate the TV proximal operator, we argue that the simpler joint approach is the more practical choice.

The self-guided DIP and splitting-based approaches employ exponential averaging to obtain the reconstruction (see line 6 in Algorithm~\ref{alg:self_guided_dip}).
Overfitting to noise is minimal, as indicated by the close values of final and peak PSNR, with the peak occurring relatively late (iteration $\num{9858}$) in the optimisation.
Exponential averaging results in a smoother PSNR curve over the iterations. In some cases this can also increases the computational time, as more iterations are necessary to obtain a suitable reconstruction. 
For all the other approaches, the peak PSNR is achieved fairly early in the optimisation process, mostly between $\num{1000}$ to $\num{1500}$ iterations.

Finally, we observe that there are only minor differences in the peak performance of the different DIP variants. This can be explained by the challenging setting of the inverse problem. In particular, we only have $m=7680$ measurements, i.e., $60$ angles and $128$ detector pixels, and the image size is $501 \times 501$ corresponding to a high undersampling. This acquisition setup violates the Nyquist sampling criterion, making accurate reconstruction feasible only in the presence of a strong prior.


\paragraph*{Role of Initialisation of Network Weights}
In the preceding experiments the DIP weights were initialised randomly.
We now compare this baseline with a warm-started variant, as discussed in Section~\ref{sec:edip}. 
To this end, the network is first pre-trained using the supervised loss in \eqref{eq:supervised_loss} on a dataset of synthetic ellipses, using the same setup as in \cite{barbano2022educated}. 
The network is trained to map noisy FBP reconstructions to their corresponding clean ground truth images, thereby learning a reconstruction prior aligned with the forward model. 
Figure~\ref{fig:edip} shows representative samples from the pre-training dataset, the FBP reconstruction of the walnut used as network input, and the initial output of the pre-trained model. 
Notably, the pre-trained DIP produces a smooth and structurally coherent reconstruction of the walnut, providing a substantially improvement, compared to the DIP with randomly initialised weights.
As shown in Figure~\ref{fig:pretrained_dip}, the warm-started DIP begins with a PSNR of approximately $24\,\mathrm{dB}$, whereas the randomly initialised DIP starts near $6\,\mathrm{dB}$. 
Although the warm-started variant attains a slightly lower peak PSNR (by roughly $0.6\,\mathrm{dB}$), its optimisation trajectory is markedly smoother and demonstrates reduced susceptibility to overfitting.
In this example, we consider only the vanilla DIP training setup without regularisation. However, all techniques aimed at improving optimisation and mitigating overfitting are also applicable to the warm-started DIP.

\begin{figure}[t]
    \centering
    \begin{subfigure}[t]{0.32\textwidth}
        \centering
        \includegraphics[width=\linewidth]{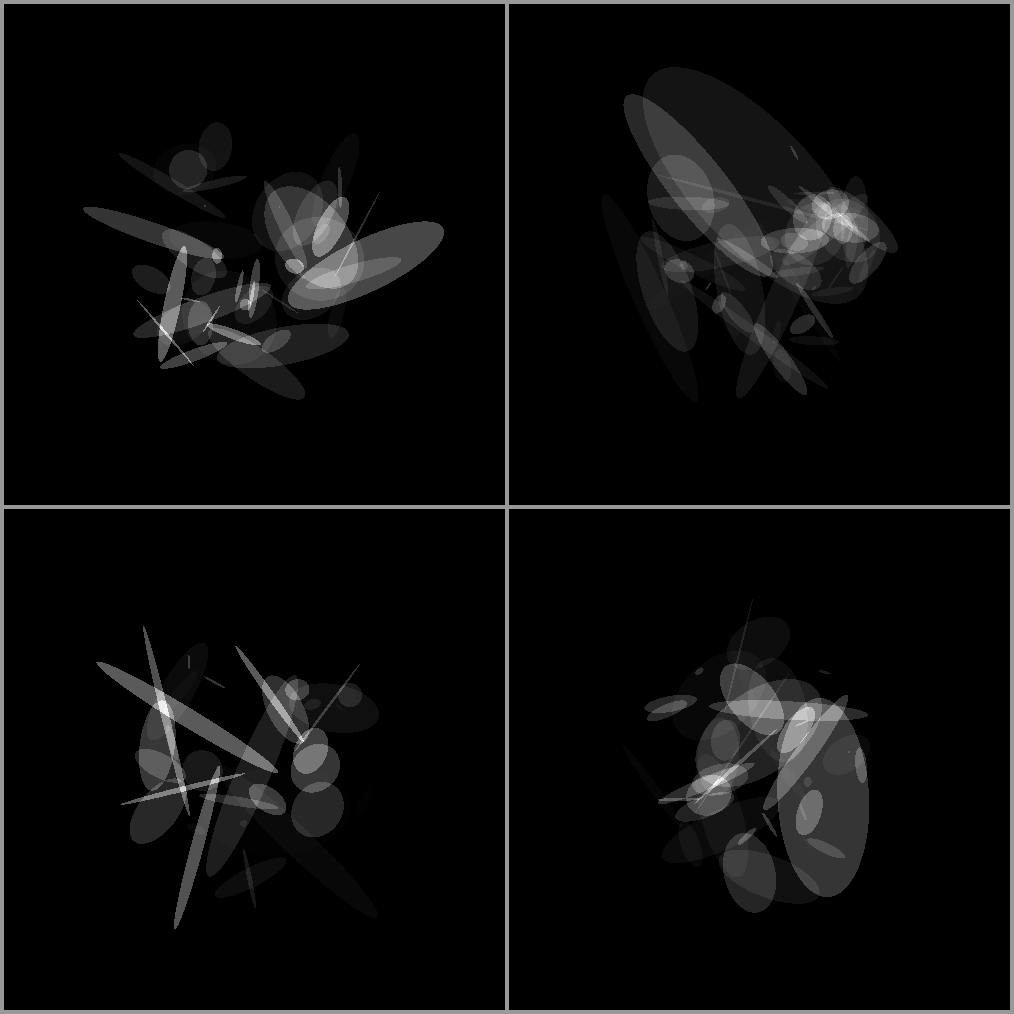}
        \caption{Training data}
    \end{subfigure}
    \hfill
    \begin{subfigure}[t]{0.32\textwidth}
        \centering
        \includegraphics[width=\linewidth]{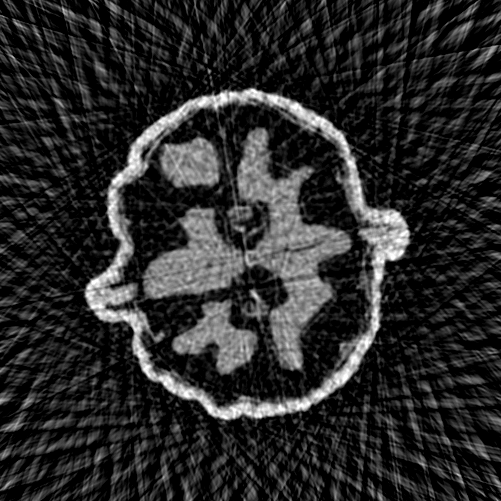}
        \caption{FBP recon.}
    \end{subfigure}
    \hfill
    \begin{subfigure}[t]{0.32\textwidth}
        \centering
        \includegraphics[width=\linewidth]{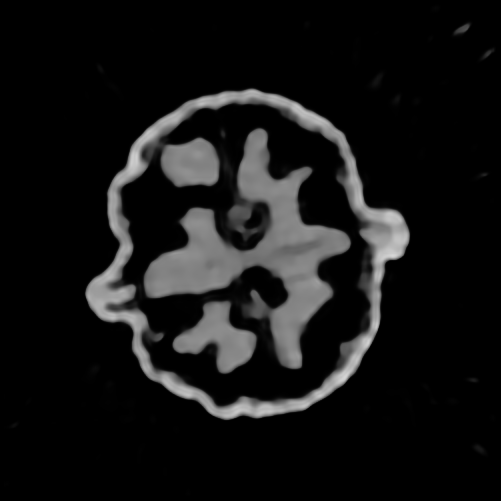}
        \caption{Initial recon.}
    \end{subfigure}
    \caption{We show \textbf{(a)} samples from the synthetic ellipses dataset used for pre-training, \textbf{(b)} the FBP reconstruction of the walnut used as input to the warm-started DIP \textbf{(c)} and the initial reconstruction of the warm-started DIP.}
    \label{fig:edip}
\end{figure}

\begin{figure}[t]
    \centering
    \includegraphics[width=0.8\linewidth]{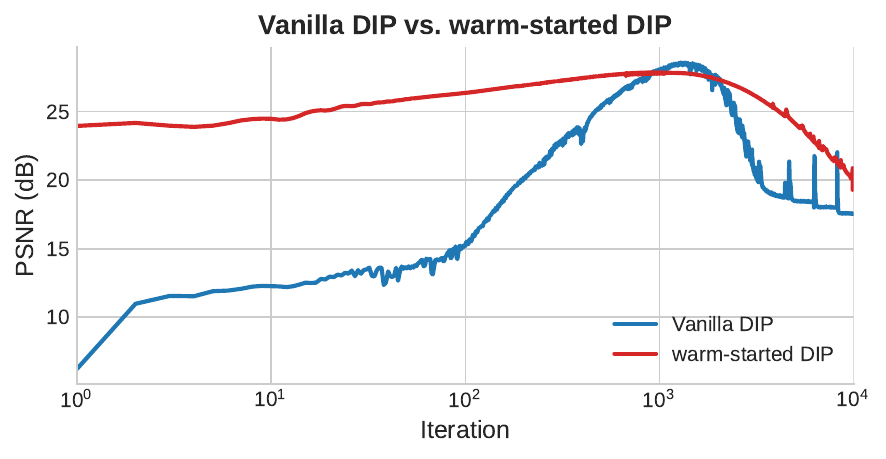}
    \caption{The evolution of the PSNR during the optimisation process for the vanilla DIP and the warm-started DIP.}
    \label{fig:pretrained_dip}
\end{figure}

\section{Conclusion and Further Work}
\label{sec:conclusion}

While our focus is on CT image reconstruction, the DIP has been used across a wide range of (medical) inverse problems, including positron emission tomography \citep{gong2018pet} or electrical impedance tomography \citep{liu2023deepeit}.
Scaling the DIP to 3D image reconstruction, where the architecture directly parametrises the volume, see \cite{singh20233d} for an application to PET, is an important extension, as the scarcity of 3D supervised datasets makes unsupervised methods like the DIP particularly relevant.

The Deep Diffusion Prior \citep{chung2024deep}, based on the steerable conditional diffusion framework  \citep{barbano2025steerable}, generalizes the DIP framework by embedding the reconstruction process within a generative diffusion model. This shifts the inductive bias from network architecture alone to the learned score dynamics of the diffusion process.

More broadly, the DIP occupies a unique position in the modern deep learning landscape of image reconstruction. At a time when the field is dominated by large-scale supervised (foundation) models trained on massive datasets, the DIP demonstrates that the architecture of a neural network alone can act as an effective regulariser. It decouples the performance from the size and choice of the dataset and instead depends on the structural properties of the underlying network architecture. 
However, the exact bias introduced by architectural choices, e.g., the network depth, the activation function or upsampling methods, remain poorly understood and lack a specific mathematical characterisation.
This contrast with classical hand-crafted regularisation function, such as TV, where the promoted features are well understood and predictable.
Furthermore, the DIP can be sensitive to early stopping rules and initialisation. 
Despite these limitations, the DIP is a promising approach, combining the advances in deep learning architectures with classical variational methods. 
Further, the DIP demonstrates that deep learning can contribute to inverse problems, even in the absence of training data. In applications such as medical imaging where data is scarce or obtaining it is expensive, the DIP remains relevant. 
Further work should aim at a more detailed theoretical understanding of the priors induced by specific network architectures. 

\section*{Acknowledgments}
Simon Arridge, Alexander Denker and Ricardo Barbano acknowledge support by the UK EPSRC programme grant EP/V026259/1.
Zeljko Kereta acknowledges support by the UK EPSRC grant EP/X010740/1.

\bibliographystyle{abbrvnat}
\bibliography{bibliography}

\end{document}